\documentclass[a4paper]{article} 
\usepackage{amsthm}
\usepackage{amsmath}
\usepackage{amsfonts}
\usepackage{amssymb}
\usepackage{a4wide}  
\usepackage[colorlinks=true,citecolor=blue, linkcolor=blue, urlcolor=blue,bookmarks=false]{hyperref}
\usepackage{bbm}
\usepackage{mathtools}
\usepackage{filecontents}

\newcommand{\bra}[1]{{\left\langle{#1}\right\vert}}
\newcommand{\ket}[1]{{\left\vert{#1}\right\rangle}}

\newcounter{counter:IQP}
\newcounter{counter:save}
\newcounter{counter:BQP}
\newcounter{counter:field}

\newcommand{\Reduce}[1]{$\leq_{\texttt {#1}}$}

\newcommand{\NIsing}[1]{{\sc FPRAS-NormIsing}\ensuremath{(#1)}}
\newcommand{\ZivIsing}[1]{{\sc ComplexApx-Ising}\ensuremath{(#1)}}
\newcommand{\NonzeroZivIsing}[1]{{\sc ComplexApx-Nonzero-Ising}\ensuremath{(#1)}}
 
\newcommand{\ZIsing}[2]{\ensuremath{Z_{\text{Ising}}(#1;#2)}}
\newcommand{\IQPIsing}[2]{{\sc Factor\/-$#1$-NormIQPIsing}\ensuremath{(#2)}}

\newcommand{\CNIsing}[2]{{\sc Factor\/-$#1$-NormIsing}\ensuremath{(#2)}}
\newcommand{\CAIsing}[2]{{\sc Distance\/-$#1$-ArgIsing}\ensuremath{(#2)}}

\newcommand{\CNonzeroNIsing}[2]{{\sc Factor\/-$#1$-Nonzero-NormIsing}\ensuremath{(#2)}}
\newcommand{\CNonzeroAIsing}[2]{{\sc Distance\/-$#1$-Nonzero-ArgIsing}\ensuremath{(#2)}}  

\newcommand{\CNonzeroNRCIsing}[2]{{\sc Factor\/-$#1$-Nonzero-Norm2Tutte}\ensuremath{(#2)}}

\newcommand{\ZTutte}[2]{\ensuremath{Z_{\text{Tutte}}\ensuremath{(#1;#2)}}}
\newcommand{\SignReTutte}[1]{{\sc Sign-RealTutte}\ensuremath{(#1)}}
\newcommand{\NonzeroSignReTutte}[1]{{\sc Sign-Real-NonzeroTutte}\ensuremath{(#1)}}

\newcommand{\SSimIQP}[2]{{\sc Factor\/-$#1$-StrongSimIQP}$_{1,2}(#2)$}

\def\maxcut{{\sc Max-Cut}}

\def\bM{\boldsymbol M}

\newcommand{\Height}[1]{{\rm H}(#1)}
\newcommand{\Mahler}[1]{{\rm M}(#1)}
\newcommand{\Res}[2]{{\rm Res}(#1,#2)}

\newcommand{\BinMatrix}[4]{\left[\begin{smallmatrix} #1 & #2 \\ #3 & #4 \end{smallmatrix}\right]}
\newcommand{\BinVec}[2]{\left[\begin{smallmatrix} #1 \\ #2\end{smallmatrix}\right]}

\let\epsilon=\varepsilon 

\def\numP{\#{\bf P}}
\def\BQP{{\bf BQP}}
\def\PP{{\bf PP}}
\def\RP{{\bf RP}}
\def\FP{{\bf FP}}
\def\NP{{\bf NP}}

\def\IQP{{\bf IQP}}
\def\IQPonetwo#1{{\bf IQP}$_{1,2}(#1)$}
\def\Pgate#1{P_{#1}}
\def\ZZgate#1{R_{#1}}

\def\pp{a}
\def\qq{b}  
\def\Pr{\mathop{\rm Pr}\nolimits}

\def\I{\mathrm{i}}

\def\CC{\mathbb{C}}
\def\RR{\mathbb{R}}
\def\ybeta{y}
\def\cst{K}  
\def\addcst{\rho}   
\def\apxnorm{\widehat{N}}
\def\apxarg{\widehat{A}}

\def\hatbgamma{\widehat{\bgamma}}
\def\bgamma{\boldsymbol\gamma}
\def\boldgamma{\bgamma}
\def\hatw{\widehat{w}}

\def\algebraics{\overline{\mathbb{Q}}}

\def\MCCut{{\sc \#Minimum Cardinality $(s,t)$-Cut}}

\makeatletter
\def\prob#1#2#3{\goodbreak\begin{list}{}{\labelwidth\z@ \itemindent-\leftmargin
                        \itemsep\z@  \topsep6\p@\@plus6\p@
                        \let\makelabel\descriptionlabel}
                \item[\bf Name]#1
                \item[\bf Instance]#2
                \item[\bf Output]#3
                \end{list}}
\makeatother
\let\epsilon=\varepsilon
\def\lgeps{R}
\def\ii{j}  
\def\jj{k}
\def\ss{s}
\def\tt{t}
\def\oo{*}

\newtheorem{theorem}{Theorem}
\newtheorem{lemma}[theorem]{Lemma}
\newtheorem{claim}[theorem]{Claim}

\newtheorem{corollary}[theorem]{Corollary}

\newtheorem{definition}[theorem]{Definition}

\newtheorem{remark}[theorem]{Remark}
\newtheorem{observation}[theorem]{Observation}

\usepackage{silence}
\begin{document}

\title{The complexity of approximating complex-valued Ising and Tutte partition functions 
\thanks{
The research leading to these results has received funding from the European Research Council under the European Union's Seventh Framework Programme (FP7/2007--2013) ERC grant agreement no.\ 334828. The paper reflects only the authors' views and not the views of the ERC or the European Commission. The European Union is not liable for any use that may be made of the information contained therein.}}
\author{Leslie Ann Goldberg \thanks{Department of Computer Science, University of Oxford, UK.}
   \and Heng Guo \thanks{School of Mathematical Sciences, Queen Mary, University of London, Mile End Road, London E1 4NS, United Kingdom.
   Part of the work was done during HG's visit in the University of Oxford.
   HG is also supported by the EPSRC grant EP/N004221/1.} }
\date{\today}

\maketitle

\begin{abstract}
We study the complexity of approximately 
evaluating the Ising and Tutte partition functions with complex parameters.
Our results are 
partly motivated by the study of the quantum complexity classes \BQP\ and \IQP.
Recent results show how to encode quantum computations as evaluations of classical
partition functions. These results rely on interesting and deep results about quantum 
computation in order to obtain hardness results about the difficulty of (classically) evaluating
the partition functions for certain fixed parameters.

The motivation for this paper is to study more comprehensively the complexity of (classically) approximating
the Ising and Tutte partition functions with complex parameters. 
Partition functions are combinatorial in nature and quantifying their approximation complexity
does not require a detailed understanding of quantum computation.
Using combinatorial arguments, we give the first full classification of the complexity of 
multiplicatively approximating the norm and additively approximating
the argument of the Ising partition function for complex edge interactions
(as well as of approximating the partition function according
to a natural complex metric).
We also study the norm approximation problem in the presence of external fields,
for which we give a complete dichotomy when the parameters are roots of unity.
Previous results were known just for a few such points, and
we strengthen these results from \BQP-hardness to \numP-hardness.
Moreover, we show that computing the sign of the Tutte polynomial is \numP-hard at certain points
related to the simulation of \BQP.
Using our classifications, we then revisit
the connections to quantum computation, drawing conclusions that are a little different
from (and incomparable to) ones in the quantum literature, but along similar lines. 
\end{abstract}

\section{Introduction}
 
We study the Ising and Tutte partition functions,
which are well-known partition functions arising in combinatorics and statistical physics (see, for example, 
\cite{Sokal}).  Early works 
which studied the  complexity of (exactly) evaluating these partition functions~\cite{JVW}
considered both real and complex parameters.
Applications in statistical mechanics actually require consideration of complex
numbers
because the possible points of physical phase transitions
occur exactly at real limit points of \emph{complex} zeroes of
these partition functions
(see Sokal's explanation in
Section 5 ``Complex Zeros of $Z_G$: Why should we care?''~\cite{Sokal}).
However, given the difficulty of completely resolving the complexity of the
approximation problem,
most works
which comprehensively studied
the complexity of \emph{approximately} evaluating 
these
partition functions \cite{GJNPhard, GJSign, JS}
restricted attention to real parameters.
A notable counter-example is the paper of Bordewich et al.~\cite{BFLW}
which studied normalised additive approximations for $\numP$ functions including
these partition functions. Bordewich et al.\ were motivated by
a result of Freedman et al.\ \cite{FKLW}
showing that an approximate evaluation of the Jones polynomial
associated with a particular complex parameter (a $5$th root of unity) 
can be used to simulate the quantum part of any algorithm in the quantum complexity class \BQP,
which is the class of decision problems solvable by a quantum computer in polynomial time
with bounded error.
The relevance of this result 
to the partition functions that we study follows from a result 
of
Thistlethwaite~\cite{Thistlethwaite}, showing that the Jones polynomial
is essentially a specialisation of the Tutte partition function.

Recently, there have been several papers  
showing how to  encode quantum computations as
evaluations of partition functions.
These results rely on interesting and deep results about quantum computation
to obtain hardness results about the difficulty of (classically) evaluating
Ising and Tutte partition functions.
For example, Kuperberg~\cite{Kuperberg} used
three results in quantum computation 
(a density theorem from~\cite{FLW}, the Solovay-Kitaev theorem (see~\cite{NC}), 
and Post\BQP=\PP\
\cite{Aaronson})
to demonstrate the $\numP$-hardness of a certain kind of approximation of 
the Jones polynomial.
His theorem is repeated later as Theorem~\ref{thm:Kup}, where it is discussed in more detail.
He also derived related results about multiplicative approximations
of the Tutte polynomial for certain real parameters.

\IQP\ stands for ``Instantaneous Quantum Polynomial time''. 
It is characterised by a  class of quantum circuits introduced by Shepherd and Bremner \cite{SB09}.
Fujii and Morimae \cite{FM} showed how to encode  
 \IQP\ circuits as instances of the Ising model.
Thus, they were able to use 
a quantum complexity result of Bremner et al.\ \cite[Corollary 3.3]{BJS11}
(showing that weakly simulating \IQP\ with multiplicative error implies that the polynomial hierarchy
collapses to the third level) to
obtain a result about the approximation of the Ising model --- namely that
an FPRAS for the Ising model with parameter $y=\exp(\I \pi/8)$  
would similarly entail collapse of the polynomial hierarchy. 
(As they mention, a similar result applies for other parameters that are universal for \IQP.)
This result is further discussed in Section~\ref{sec:IQP}.
Other examples include \cite[Result 2]{DDVM11},
\cite[Theorem 6.1]{ICKB},
and \cite[Theorems 2 and 3]{MFI}
which give \BQP-hardness of certain Ising model approximations, enabling the conclusion that
certain efficient algorithms for approximating these partition function up to additive error are unlikely to exist.
Ilblisder et al.~\cite{ICKB} point out that some instances
that they prove hard \emph{do} have multiplicative approximations, due to Jerrum and  Sinclair~\cite{JS},
emphasising the difference between additive and multiplicative approximation.
Matsuo et al.~\cite[Theorem 4]{MFI} also relate the simulation of \IQP\ circuits
to Ising model approximations with real parameters.
  
The  
motivation for our paper is to 
study more comprehensively the complexity of
approximating the Ising partition function at complex parameters, and
also to
go the other way around, working from the combinatorial model
to quantum computation.
Partition functions are combinatorial in nature
and classifying the difficulty of approximating these partition functions should not require
a detailed understanding of quantum mechanics or quantum computation.
Hence, we undertake a detailed classification of 
the complexity of the partition function problems, 
using combinatorial methods. 
We focus mainly on the Ising model 
since this model is particularly relevant in statistical physics
(Section \ref{sec:hardIsing}).
This model is also connected to \IQP\ (as explained in Section~\ref{sec:IQP}).
We also consider  
the more general Tutte polynomial at any point $(x,y)$ 
where $x=-t$ and $y=-{t}^{-1}$ for a root of unity $t$
(this is connected to \BQP, as will be explained in Section~\ref{sec:BQP}).

Our main result for the Ising model (Theorem~\ref{thm:main}) is a classification of the complexity of
approximating the partition function with complex edge interactions.
This result is illustrated in Figure~\ref{fig:complexplane}.
As the figure shows, there are very few parameters (edge interactions) in the complex plane
for which the approximation problem is tractable.
For most edge interactions, it is extremely intractable ($\numP$-hard to
approximate the norm within any constant factor and to approximate the argument
within $\pm \pi/3$).  Theorem~\ref{thm:relaxed} extends these results to
a more relaxed setting in which approximation algorithms are unconstrained (allowed to
output any rational number) if the correct output is zero.
We emphasise that the goal of our work is to classify the 
difficulty of the problem for all \emph{fixed} parameters in the
complex plane.
The proofs of our theorems are elementary and combinatorial.
The main idea (see  Lemma~\ref{lem:bisect}) is an extension of a bisection technique of Goldberg and Jerrum~\cite{GJSign}
showing how to use an approximation for the norm of a function to get
very close to a zero of the function.
Our result for the Tutte polynomial (\ref{thm:BQP}) is also proved using bisection.
It shows that, for 
\emph{any}
relevant parameters, it is $\numP$-hard to determine
whether the sign of the polynomial is non-negative or non-positive (with an arbitrary answer being
allowed when it is zero). 

Using our classifications, we
then revisit the connections to quantum computation,
drawing conclusions that are a little different from (and incomparable to) the ones 
in the papers mentioned earlier,
but along similar lines, as we now explain.
Theorem~\ref{thm:IQP:hardness} shows that strong simulation of \IQP\
within any constant factor is $\numP$-hard, even for the restricted class of
circuits considered by Bremner et al.~\cite{BJS11}.
Our result is incomparable to their hardness result \cite[Corollary 3.3]{BJS11}.
Both results show hardness of multiplicative approximation. However,
their result is for weak simulation (sampling from the output distribution of
the circuit) whereas ours is for strong simulation (estimating
the probability of a given output). In general, hardness results about weak simulation
are more desirable, however multiplicative approximation is less
appropriate for weak simulation, where total variation distance
is more important. Also, our results (unlike those of~\cite{BJS11}) are not
sensitive to the behaviour of the algorithms when the correct value is zero.
Moreover, our complexity assumption (that $\FP\neq \numP$)
is implied by and therefore milder
than theirs 
(that the polynomial hierarchy does not collapse to the third level).
These results are discussed further in Section~\ref{sec:IQP}.

It seems that a result similar to our \IQP\ result could also be 
obtained via Boson sampling~\cite{BosonSampling}.
In particular, Aaronson and Arkhipov~\cite[Theorem 4.3]{BosonSampling}
have used a bisection technique similar to 
the one of Goldberg and Jerrum~\cite{GJSign}
to show that approximating the square of the permanent of 
a real-valued input matrix within a constant factor is \numP-hard.
Any such input \cite[Lemma 4.4]{BosonSampling} can be turned into a unitary matrix which can be viewed
as a ``Boson Sampling'' input. The output of the Boson sampling problem
is essentially the square of the permanent of the matrix (so is hard to approximate).
Furthermore, 
the Boson sampling problem
can be simulated by \BQP\ circuits and adaptive \IQP\ circuits
(in the strong sense).
Thus, while it is interesting to see that our Ising-model
results have \IQP\ applications, 
the important point concerning our result is the
comprehensive classification of the Ising complexity, 
rather than the particular quantum applications.

As we explain in Section~\ref{BQPintro}, classical simulation of the complexity class \BQP\
is related to (but not directly a consequence of) determining the sign of the Tutte polynomial 
at a certain point $(-t,-t^{-1})$. Theorem~\ref{thm:BQP} shows that this problem is $\numP$-hard
(even when the algorithm is not required to handle the case in which the output is zero),
answering a question raised by Bordewich et al.~\cite{BFLW}.
This is related to (but incomparable to) a result (Theorem~\ref{thm:Kup}) of Kuperberg~\cite{Kuperberg}. 
These results are discussed further in Section~\ref{sec:BQP}.

Finally, we study Ising models with external fields.
De las Cuevas et al.\ \cite[Result 2]{DDVM11} showed that with edge interaction $\I$ and external field $e^{\I\pi/4}$ 
an additive approximation of the partition function is \BQP-hard.
Motivated by such connections, we focus on the 
problem of (multiplicatively) approximating the norm of the partition function
when both the interaction parameter and the external field are roots of unity.
We extend our hardness results to show that, 
for most such parameters, including the one studied by De las Cuevas et al.,
the approximation problem is \numP-hard
(for an exact statement, see Theorem~\ref{thm:fields}).
For the remaining parameters,
the partition function can be evaluated exactly in polynomial time,
and thus we get a complete dichotomy (Theorem~\ref{thm:fields}).
This extension relies on some lower bounds from transcendental number theory, which allow us to
convert additive distances into multiplicative ones. 
The lower bound results are given in Section~\ref{sec:LB} and 
our hardness results are in Section~\ref{sec:field}.

As we have already mentioned, there are many papers 
encoding quantum simulations as Ising models, including especially the result of Fujii and Morimae~\cite{FM}.
We could use this encoding (along with our Theorem~\ref{thm:relaxed}) to derive
our  quantum application (Theorem~\ref{thm:IQP:hardness}).
In order to make the paper self-contained, and to make it accessible to readers  from outside
the area of quantum computation we instead give our own, more combinatorial, presentation
of how to encode \IQP\ circuits as Ising instances.  This is given in Section~\ref{sec:IQP}.

\subsection{The Ising model}

The main partition function that we study is the partition function of the Ising model.
Let $\ybeta$ (called the edge \emph{interaction}) and $\lambda$ (called the \emph{external field}) be two parameters.
The partition function is defined for a (multi)graph $G=(V,E)$ as 
\begin{align}
  \label{eqn:Ising}
  \ZIsing{G}{\ybeta,\lambda}=\sum_{\sigma:V\rightarrow\{0,1\}}\ybeta^{m(\sigma)}\lambda^{n_1(\sigma)},
\end{align}
where $m(\sigma)$ is the number of monochromatic edges under $\sigma$ 
(that is, the number of edges $(u,v)$ with $\sigma(u)=\sigma(v)$)
and $n_1(\sigma)$ is the number of vertices~$v$ with $\sigma(v)=1$.  
We write $\ZIsing{G}{\ybeta}$ to denote $\ZIsing{G}{\ybeta,1}$. 
   
We will consider complex parameters~$\ybeta$ and~$\lambda$
from the set  $\algebraics$ of algebraic numbers. 
Thus, the real and imaginary parts of~$\ybeta$ and~$\lambda$ will be algebraic.
We use $\arg(z)$ to denote the arg of a complex number $z$.
For fixed $\ybeta$ and $\lambda$,
we study  several computational problems.
The first of them is approximating the norm of $\ZIsing{G}{\ybeta,\lambda}$
within a factor $\cst>1$.

\prob{\CNIsing{\cst}{\ybeta,\lambda}.}
{A (multi)graph $G$.}
{A  rational number $\apxnorm$ such that
$  {\apxnorm}/{\cst} \leq  |\ZIsing{G}{\ybeta,\lambda}| \leq  \cst \apxnorm$.}

 We also consider  the problem of approximating the argument of the partition function
 within an additive distance of $\addcst \in (0,2 \pi)$.
 Here we have to treat the zero case exceptionally since the argument is undefined. 
  
\prob{\CAIsing{\addcst}{\ybeta,\lambda}.}
{ A (multi)graph $G$.}{ 
If $\ZIsing{G}{\ybeta,\lambda}=0$, then~$0$. Otherwise,
a  rational number $\apxarg$ such that 
$$|\apxarg -  \arg(\ZIsing{G}{\ybeta,\lambda})| \leq  \addcst.$$}

We drop the argument~$\lambda$ when it is equal to~$1$, so
\CNIsing{\cst}{\ybeta}  denotes the problem \CNIsing{\cst}{\ybeta,1 }
and   \CAIsing{\addcst}{\ybeta} denotes \CAIsing{\addcst}{\ybeta,1}.

\subsection{Approximating Complex Numbers}
  \label{sec:Ziv}

 It makes sense that we 
 approximate the norm of a complex number relatively, whereas we approximate
 the argument additively.
 This is natural because multiplying complex numbers multiplies norms and adds arguments,
 so it preserves the usual property that if you can approximate two numbers, you can approximate
 the product.
 
 Other notions of approximation have  been proposed. 
 Most notably,
 Ziv~\cite{Ziv}
 has proposed that the distance between two complex numbers~$y$ and
  $ y'$  
 should be 
 measured as 
$$d( y',y) = \frac{| y'-y|}{\max(| y'|,|y|)},$$
 where $d(0,0)=0$.
  We  also study the following approximation problem.
 
  \prob {\ZivIsing{\ybeta,\lambda}.}{
 A (multi)graph $G$ and a positive integer $\lgeps$, in unary.}
{If $|\ZIsing{G}{\ybeta,\lambda}|=0$ then the algorithm should output~$0$. Otherwise, it should
output a
 complex number $y$ such that
  $d(y, \ZIsing{G}{\ybeta,\lambda}) \leq  \tfrac{1}{\lgeps}$.
 }

 As with the other problems, we use the notation \ZivIsing{\ybeta} for \ZivIsing{\ybeta,1}.
 We have specified the error~$\lgeps$ as an input of the problem, rather than as a parameter
 in order to emphasise the suitability of \ZivIsing{\ybeta,\lambda} as
 an appropriate notion of approximation for the Ising partition function when $y$ is complex. The number $\lgeps$ is
 expressed in unary so   a polynomial time algorithm for
 \ZivIsing{\ybeta,\lambda} 
 would give a so-called ``fully polynomial time approximation scheme'' for
 the norm of the partition function.
 For partition functions, it is well-known that approximating the norm within 
 a factor that is an inverse polynomial in a unary input~$\lgeps$
   is equivalent in difficultly to approximating the norm with
 any specific factor~$\cst>1$.
 We will return to this point later in Lemma~\ref{lem:duplicate}.

\subsection{Main results for the Ising model}
\label{sec:mainIsing}

\begin{figure}
  \label{fig:complexplane}
  \centering
    \includegraphics{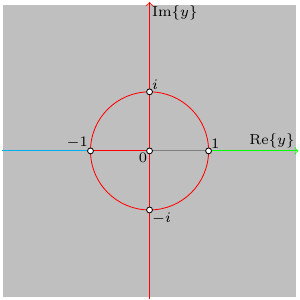}
    \caption{An illustration of Theorem~\ref{thm:main} for \CNIsing{\cst}{ y }. 
    The five white points correspond to the easy 
    evaluations described in Item~1.
    The green line segment $(1,\infty)$ corresponds to a region where approximation is in \RP --- See Item~2.
    The blue line segment $(-\infty,-1)$ corresponds to a region where approximation is equivalent to
    approximately counting perfect matchings. See Item~4.
    The red points on the axes (the imaginary axis and the segment $(-1,0)$) and on the unit circle correspond to regions where
    approximation is $\numP$-hard. See Items~5, 6, and 7.
    Elsewhere the points are coloured grey, and approximation is known to be NP-hard (Items~3, 9 and 10)
    and sometimes to be $\numP$-hard (Item~8, not pictured). }
\end{figure}

The following theorem gives our main complexity results about the Ising model.
These results classify the problem of approximating the partition function over the entire complex plane. 
For every value of the parameter $y$, we either show that the problem is easy, 
in the sense that both the norm and the arg of the partition function can be well-approximated
(and so can the conglomerate problem using the Ziv distance),
or we show that approximating at least one these is hard
(and so is the conglomerate problem using the Ziv distance).
The results for approximation of the norm are illustrated in Figure~\ref{fig:complexplane}.

\begin{theorem}
\label{thm:main}
Let $y=r e^{i  \theta}$ be an algebraic complex number with $r\ge 0$ and $\theta\in [0,2\pi)$.
Suppose $\cst>1$.
\begin{enumerate}
\item \label{itemexact} If $y=0$ or if
$r=1$ and $\theta \in \{0,\frac{\pi}{2},\pi,\frac{3\pi}{2}\}$
then \CNIsing{\cst}{ y }, \CAIsing{(\pi/3)}{ y} and \ZivIsing{y} 
are in \FP.
  
\item \label{realtwo} If $y>1$ is a real number then \CNIsing{\cst}{ y} 
and \ZivIsing{y} are
in \RP\ 
and  \CAIsing{(\pi/3)}{ y} is in \FP.

\item \label{realthree} If $y$ is a real number  in $(0,1)$ then  \CNIsing{\cst}{ y} 
and \ZivIsing{y} are
\NP-hard
and  \CAIsing{(\pi/3)}{ y} is in \FP.

\item \label{realfour} If $y<-1$ is a real number then \CNIsing{\cst}{ y} is equivalent in complexity to
the problem of approximately counting perfect matchings in graphs
and \ZivIsing{y} is as hard.
However, 
\CAIsing{(\pi/3)}{ y} is in \FP.

\item \label{realfive} If $y$ is a real number in $(-1,0)$ then \CNIsing{ \cst}{ y },
\CAIsing{ (\pi/3)}{y}   
and \ZivIsing{y}
are $\numP$-hard.

\item \label{circleone} If $r=1$ and $\theta\not\in\{0,\frac{\pi}{2},\pi,\frac{3\pi}{2}\}$
then \CNIsing{\cst}{ y }, \CAIsing{(\pi/3)}{ y} and \ZivIsing{ y} are \numP-hard. 
\item \label{imaginaries} If $\theta \in \{ \frac{\pi}{2}, \frac{3\pi}{2}\}$
and $r\not\in\{-1,0,1\}$ then \CNIsing{\cst}{ y},
\CAIsing{(\pi/3)}{ y} and \ZivIsing{ y} are \numP-hard. 
\item \label{coprime} If $r>0$ and $\theta=\frac{\pp\pi}{2\qq}$, 
  where $\pp$ and $\qq$ are two co-prime positive integers and $\pp$ is odd then 
  \CNIsing{\cst}{ y }, \CAIsing{(\pi/3)}{ y} and \ZivIsing{ y} are \numP-hard.   
\item \label{NPitem}
If $r<1$ and $y\neq 0$ then    \CNIsing{\cst}{ y}  and  \ZivIsing{ y} are \NP-hard.
 \item \label{NPlastitem}
 If $r>1$ and $\theta\not\in\{0,\pi\}$ then  \CNIsing{\cst}{ y}  and  \ZivIsing{ y} are \NP-hard.
   \end{enumerate}

\end{theorem}

\subsection{Relaxed versions of the problems}

 A polynomial-time algorithm for  any of the problems that
 we have defined is required to output $0$ 
if it is given an input~$G$ such that
 $\ZIsing{G}{\ybeta,\lambda}=0$. Theorem~\ref{thm:main} gives hardness 
 results for 
  these problems.
 The hardness    is not due to special 
difficulties which arise when the value of the partition function is zero.
In order to demonstrate this point,
(and in order to make certain reductions easier later on),
we also consider the following, more relaxed versions of the problems, where the
output is unconstrained if the value of the partition function is zero. As before,
the parameter $\cst$ is greater than~$1$ and the parameter~$\addcst$ is in $(0,2 \pi)$.
 
\prob {\CNonzeroNIsing{\cst}{\ybeta,\lambda }.}
{A (multi)graph $G$.}
{If $|\ZIsing{G}{\ybeta,\lambda}|=0$ then the algorithm may output
any rational number. Otherwise, it must output
a  rational number $\apxnorm$
such that
$  {\apxnorm}/{\cst} \leq  |\ZIsing{G}{\ybeta,\lambda}| \leq  \cst \apxnorm$. }

\prob{ \CNonzeroAIsing{\addcst}{\ybeta,\lambda}.}{ A (multi)graph $G$.}{ 
If $\ZIsing{G}{\ybeta,\lambda}=0$, then the algorithm may output any rational number.
Otherwise, it must output a  rational number $\apxarg$ such that 
$|\apxarg -  \arg(\ZIsing{G}{\ybeta,\lambda})| \leq  \addcst$.}

\prob {\NonzeroZivIsing{\ybeta,\lambda}.}{
 A (multi)graph $G$ and a positive integer $\lgeps$, in unary.}
{If $|\ZIsing{G}{\ybeta,\lambda}|=0$ then the algorithm may output any complex number. Otherwise,
it must output a 
complex number $z$ such that
  $d(z, \ZIsing{G}{\ybeta,\lambda}) \leq  \tfrac{1}{\lgeps}$.
 } 

As in the un-relaxed versions of the problems, we drop the parameter ``$\lambda$'' from the
problem name when it is~$1$. We give the following generalisation of
Theorem~\ref{thm:main}.

\begin{theorem}
\label{thm:relaxed}
All of the results in Theorem~\ref{thm:main} 
extend to the relaxed case.
That is, the results are still  true with
\CNIsing{\cst}{ y }, \CAIsing{(\pi/3)}{ y} and \ZivIsing{y} 
replaced by 
\CNonzeroNIsing{\cst}{ y }, \CNonzeroAIsing{(\pi/3)}{ y} and \NonzeroZivIsing{y}, respectively.

\end{theorem}

\subsection{Applications to quantum simulation}

\subsubsection{\IQP}
\IQP\ is characterised by a restricted class of quantum circuits~\cite{SB09}.
We will  give a formal definition    in Section~\ref{sec:IQP}.
There we will also discuss related work by 
Fujii and Morimae \cite{FM}, Bremner et al.\ \cite{BJS11}
and Jozsa et al.~\cite{JN14}. 
Here we give an informal description that enables us to state our theorem.
Bremner et al.\ showed a hardness of a certain kind of ``weak simulation'' of
a restricted class of circuits called \IQPonetwo{\theta}\ circuits (see Definition~\ref{def:IQPonetwo}).
The qubits of the circuit travel along ``lines'' which go into (and out of) quantum gates.
The output of such a circuit~$C$ is a random variable~$\mathbf{Y}$ (over the qubits that get measured in the output).
Given as input the string of all zero qubits  and an output string $\mathbf{y}\in\{0,1\}^{|I|}$ 
on a set $I$ --- the set of qubits that are measured in the output,
$\Pr_{C;I}$ denotes the probability that $\mathbf{Y}=\mathbf{y}$.
Strong simulation is the problem of (approximately) computing this probability.
We consider the following problem where $\cst>1$ is an error parameter.

\prob{\SSimIQP{\cst}{\theta }.}
{An \IQPonetwo{\theta} circuit $C$, a subset $I\subseteq [n]$ of lines, and a string $\mathbf{y}\in\{0,1\}^{|I|}$.}
{A rational number $p$ such that 
  $	 p/\cst \leq \Pr_{C;I}(\mathbf{Y}=\mathbf{y}) \leq  \cst p$.
}

Our main result  regarding this application is the following.
\setcounter{counter:IQP}{\value{theorem}}
\begin{theorem}
  \label{thm:IQP:hardness}
  Suppose $\cst>1$ and $\theta\in(0,2\pi)$.
  If  $e^{\I \theta}$ is an algebraic complex number
  and $e^{\I 8 \theta} \neq 1$ then  
  \SSimIQP{\cst}{\theta }\ is \numP-hard.
\end{theorem}

\subsubsection{Connections between the Sign of the Tutte Polynomial and \BQP}
\label{BQPintro}

The partition function $\ZIsing{G}{\ybeta,\lambda}$ is
equivalent to a specialisation of the \emph{Tutte polynomial},
which is a graph polynomial with two parameters, $x$ and $y$, defined as follows, 
\begin{align}
  \label{eqn:Tutte}
   T(G;x,y)=\sum_{A\subseteq E(G)}(x-1)^{\kappa(A)-\kappa(E(G))}
  (y-1)^{|A|-n+\kappa(A)},
\end{align}
where $n=|V(G)|$ and
$\kappa(A)$ is the number of connected components in the subgraph $(V(G),A)$.
If the quantity $q=(x-1)(y-1)$ is a positive integer,
then the Tutte polynomial with parameters~$x$ and~$y$
is closely related to the partition function of the Potts model, which includes
the Ising model as the special case $q=2$.
In particular, when $q=2$,
\begin{equation}
\label{eq:IsingTutte}
T(G;x,y)= {(y-1)}^{-n} {(x-1)}^{-\kappa(E(G))} \ZIsing{G}{y}.
\end{equation}

Bordewich et al.~\cite{BFLW} raised
the question ``of determining whether the Tutte polynomial is
greater than or equal to, or less than zero at a given point.''
As we will see, this  question is relevant
to the quantum complexity class~\BQP.  We consider the following problems.

\prob{\SignReTutte{x,y}}
{ A (multi)graph $G$.}
{Determine whether the sign of the real part of $T(G;x,y)$ is positive, negative, or~$0$.}

\prob{\NonzeroSignReTutte{x,y}}
{ A (multi)graph $G$.}
{A correct statement of the form ``$T(G;x,y) \geq 0$'' or
``$T(G;x,y) \leq 0$''. }

\BQP\ is  the class of decision problems solvable by a quantum computer in polynomial time with bounded error.
The theorem  \cite[Theorem~6.1]{BFLW}     shows that
all of the problems in  \BQP\  can also be solved
classically in polynomial time using an oracle
that returns the sign of the real part of the Jones polynomial of a link,
evaluated at the point~$t=\exp(2 \pi i/5)$. Thistlethwaite~\cite{Thistlethwaite}  
(see \cite[(6.1)]{JVW}), showed
that this problem is, in turn, related to the problem
of evaluating the Tutte polynomial 
$T(G;-t,-t^{-1})$, for a planar graph~$G$.
This inspired  the question of Bordewich et al.\
about the complexity of determining the sign of the Tutte polynomial,
particularly for the point
$(x,y)=(-t,-t^{-1})$.
We show that problem is hard for values of~$t$ including the relevant value $t=\exp(2 \pi i/5)$. 
Note that our result does not have direct implications for the
simulation of \BQP\ because we do not deal with planarity (though it does answer
the question of Bordewich et al.).
We give the details in Section~\ref{sec:BQP}, where we also discuss a related result of Kuperberg~\cite{Kuperberg}.
Our theorem is as follows.

\setcounter{counter:BQP}{\value{theorem}}
\begin{theorem}
\label{thm:BQP}
Consider the point $(x,y) = (\exp(-a \pi \I/b),\exp(a \pi \I/b))$,
where $a$ and $b$ are positive integers 
satisfying $0< a/b < 2$  and $a\not\in\{b/2,b,3b/2\}$.
If $a$ is odd and  $\cos(a \pi/b) < 11/27$
then \NonzeroSignReTutte{x,y} is $\numP$-hard. 
Thus \SignReTutte{x,y} is also $\numP$-hard.
\end{theorem}

The condition  $\cos(a \pi /b)< 11/27$ is  roughly  $ 0.36643 < a/b < 1.63357$.
Since $-\exp(-2 \pi \I/5) = \exp(\pi \I) \exp(-2 \pi \I/5) = \exp(3 \pi \I/5)$, we get the
 relevant corollary by taking $a=3$ and $b=5$.
 
\begin{corollary}
\label{cor:five}
  Let $y=-\exp(-2 \pi \I/5)$.  Then \NonzeroSignReTutte{1/y,y} is $\numP$-hard.
\end{corollary}

\subsection{Results about Ising models with fields}

Our results in Section~\ref{sec:mainIsing} are about the complexity of evaluating the Ising partition 
function in the absence of an external field (when $\lambda=1$). 
This is appropriate for the application to \IQP. 
Ising models with external fields are important for their own sake.
Moreover, De las Cuevas et al.\ \cite[Result 2]{DDVM11} showed that with edge interaction $\I$ and external field $e^{\I\pi/4}$ 
an additive approximation of the partition function is \BQP-hard.
Motivated by such quantum connections, we give the following extension. 
\setcounter{counter:field}{\value{theorem}}
\begin{theorem}\label{thm:fields}
 Let $\cst>1$.
 Let $\ybeta$ and $z$ be two roots of unity.
 Then the following holds:
 \begin{enumerate}
   \item If $\ybeta=\pm\I$ and $z\in\{1,-1,\I,-\I\}$, or  $\ybeta=\pm 1$, then \ZIsing{-}{\ybeta,z} can be computed exactly in polynomial time.
   \item Otherwise \CNonzeroNIsing{\cst}{\ybeta,z} is \numP-hard.
 \end{enumerate}
\end{theorem}

\section{Preliminaries}

\subsection{Facts about Approximating Complex Numbers}

  We will use the following technical lemma 
  concerning Ziv's distance measure from Section~\ref{sec:Ziv}.

\begin{lemma}\label{lem:ziv}
If $z$ and $z'$ are two non-zero complex numbers and
if $d( z',z) \leq \epsilon$ then
$| z'|/|z| \leq 1/(1-\epsilon)$ and
$ |\arg{z} - \arg{ z'}| \leq \sqrt{36 \epsilon/11}$. \end{lemma}
\begin{proof}
Suppose $d(z',z) \leq \epsilon$ and
$| z'| \geq |z|$.

First, by the triangle inequality, $|z| + |z'-z| \geq |z'|$ so
 $$
\frac{|z'|}{|z|}  = 1 + \frac{|z'|-|z|}{|z|}
\leq 1 + \frac{|z'-z|}{|z|}
= 1 + \frac{|z'-z|}{|z'|}\frac{|z'|}{|z|}
\leq 1 + \epsilon \frac{|z'|}{|z|},$$
as required.

Second, $|z'-z| \le \epsilon |z'|$ so $( | z'-z | ) ^2 \le \epsilon^2 |z'|^2$.
Letting $z = r \exp(i \theta)$ and $z' = r' \exp(i \theta')$
we have
$$ {((r'\cos(\theta')-r \cos(\theta))}^2 
+ {((r'\sin(\theta') - r \sin(\theta))}^2 \leq \epsilon^2 {r'}^2.$$
The left-hand-side is equal to 
${r}^2 + {r'}^2 - 2 r r'\cos(\theta-\theta')$.
But we already proved
$$1 \leq \frac{r'}{r} \leq \frac{1}{1-\epsilon},$$
so 
$${r'}^2{(1-\epsilon)}^2 + {r'}^2 - 2 {r'}^2 \cos(\theta-\theta')\leq \epsilon^2 {r'}^2,$$
which implies, by re-arranging the above,
$$\cos(\theta-\theta') \geq 1 - \frac{3\epsilon}{2} + \frac{\epsilon^2}{2}.$$
But 
$\cos(x ) = 1-x^2/2! + x^4/4! - x^6/6! + \cdots$, so  
$$\frac{{(\theta-\theta')}^2}{2!} - \frac{{(\theta-\theta')}^4}{4!}  + \frac{{(\theta-\theta')}^6}{6!} - \cdots \leq \frac{3 \epsilon}{2} - \frac{\epsilon^2}{2}.$$
Provided that $\epsilon$ is sufficiently small (so $\theta-\theta'\leq 1$)
the left-hand-side is at least  $\frac{{(\theta-\theta')}^2}{2!} - \frac{{(\theta-\theta')}^4}{4!} $
which is equal to 
${11(\theta-\theta')}^2/24$, so
$|\theta-\theta'| \leq \sqrt{36 \epsilon/11}$.
\end{proof}

 \begin{lemma}
 \label{lem:zivapprox} 
 Suppose $\cst>1$ and $0<\addcst < 2 \pi$. Then
 the following polynomial-time Turing reductions exist.
 \begin{gather*}
  \text{\CNIsing{\cst}{\ybeta,\lambda}} \leq_{\text{T}} \text{\ZivIsing{\ybeta,\lambda}},\\
  \text{\CNonzeroNIsing{\cst}{\ybeta,\lambda}} \leq_{\text{T}} \text{\NonzeroZivIsing{\ybeta,\lambda}},\\
  \text{\CAIsing{\addcst}{\ybeta,\lambda}} \leq_{\text{T}} \text{\ZivIsing{\ybeta,\lambda}},\\
    \text{\CNonzeroAIsing{\addcst}{\ybeta,\lambda}} \leq_{\text{T}} \text{\NonzeroZivIsing{\ybeta,\lambda}},
    \end{gather*}
    
 \end{lemma}
 \begin{proof}    
  Let $\lgeps$ be any (sufficiently large) integer so that $1-1/\lgeps > 1/\cst$
 and $\sqrt{36/11\lgeps} \leq \addcst$.
 
 Consider a  multigraph $G$ 
 where $|\ZIsing{G}{\ybeta,\lambda}|\neq 0$.
 Given input~$G$ and~$\lgeps$, an oracle for \ZivIsing{\ybeta,\lambda}
or 
 \NonzeroZivIsing{\ybeta,\lambda}
 returns a complex number~$z$ such that  $d(z, \ZIsing{G}{\ybeta,\lambda}) \leq  \tfrac{1}{\lgeps}$.
On the other hand, if $|\ZIsing{G}{\ybeta,\lambda}|= 0$,
then the oracle for \ZivIsing{\ybeta,\lambda} returns the complex number $z=0$
and the oracle for \NonzeroZivIsing{\ybeta,\lambda}
returns any complex number~$z$.

For the first two reductions,
suppose first that $|\ZIsing{G}{\ybeta,\lambda}|\neq 0$.
Then by Lemma~\ref{lem:ziv},  $d(z, \ZIsing{G}{\ybeta,\lambda}) \leq  \tfrac{1}{\lgeps}$
implies  
$$\frac{|z|}{\cst} \leq \left(1-\frac{1}{\lgeps}\right)|z| \leq  { |\ZIsing{G}{\ybeta,\lambda}|}  \leq \frac{|z|}{1-\frac{1}{\lgeps}} \leq \cst |z|,$$
so $|z|$ is a suitable output to \CNIsing{\cst}{\ybeta,\lambda} 
or  \CNonzeroNIsing{\cst}{\ybeta,\lambda} with input~$G$.
On the other hand, if $|\ZIsing{G}{\ybeta,\lambda}|= 0$
then $|z|$ is still suitable in both cases.

For the  last two reductions,
suppose first that $|\ZIsing{G}{\ybeta,\lambda}|\neq 0$.
Then by Lemma~\ref{lem:ziv}, 
$d(z, \ZIsing{G}{\ybeta,\lambda}) \leq  \tfrac{1}{\lgeps}$
implies   
$$  |\arg{z} - \arg{ \ZIsing{G}{\ybeta,\lambda} }| \leq \sqrt{36 \epsilon/11} \leq \addcst,$$
so $\arg{z}$ is a suitable output to  \CAIsing{\addcst}{\ybeta,\lambda}
or \CNonzeroAIsing{\addcst}{\ybeta,\lambda} with input~$G$. On the other hand, if $|\ZIsing{G}{\ybeta,\lambda}|= 0$
 and $z=0$ then $0$ is a suitable output in both cases.
     If $|\ZIsing{G}{\ybeta,\lambda}|= 0$ and $z\neq 0$
 then $\arg z$ is suitable (as an output for \CNonzeroAIsing{\addcst}{\ybeta,\lambda}).

\end{proof}

\subsection{The multivariate Tutte polynomial}
  We will require the random cluster formulation of the 
  multivariate Tutte polynomial.
  Given a (multi) graph $G$ with edge weights $\boldgamma: E(G) \rightarrow 
  \algebraics$ and $q\in \algebraics$, this is defined as
   \begin{align}
	\label{eqn:RCTutte}
	\ZTutte{G}{q,\boldgamma}:=\sum_{A\subseteq E(G)}q^{\kappa(A)}
	\prod_{e\in A} \gamma_e,
  \end{align} 
  where $\gamma_e$ is a shorthand for $\boldgamma(e)$ for an edge $e\in E(G)$.

  Suppose $x$ and $y$  satisfy $q=(x-1)(y-1)$.
  For a graph $G=(V,E)$, let $\boldgamma:E \rightarrow  \algebraics$ be
  the constant function which maps every edge to the value $y-1$.
  Then   (see, for example \cite[(2.26)]{Sokal})
  \begin{equation}
  \label{eq:TutteRC} 
  T(G;x,y)= {(y-1)}^{-n} {(x-1)}^{-\kappa(E(G))}  \ZTutte{G}{q,\boldgamma}.
  \end{equation}
  Obviously from \eqref{eq:IsingTutte}, this implies that if $q=2$   then
  $\ZIsing{G}{y} = \ZTutte{G}{q,\boldgamma}$.
   
  To apply a technique from~\cite{GJSign}
  we will require a multivariate version of the problem 
  \CNonzeroNIsing{\cst}{\ybeta,\lambda }.
  We could do this for general~$q$,
  but we will only use the following version, which is restricted to $q=2$
  and has two complex parameters, $\gamma_1$ and $\gamma_2$.
   
  \prob {\CNonzeroNRCIsing{\cst}{\gamma_1,\gamma_2 }.}
  {A (multi)graph $G=(V,E)$ and edge weights
  $\boldgamma: E \rightarrow \{\gamma_1,\gamma_2\}.$}
  {If $|\ZTutte{G}{2,\boldgamma}|=0$ then the algorithm may output
  any rational number. Otherwise, it should output
  a  rational number $\apxnorm$
  such that
  $  \apxnorm/\cst \leq  |\ZTutte{G}{2,\boldgamma}| \leq   \cst \apxnorm$.}
   
  Suppose that $s$ and $t$ are two distinguished vertices of~$G$.
  Let $Z_{st}(G;q,\boldgamma)$ be the contribution to $\ZTutte{G}{q,\boldgamma}$ from 
  subgraphs $A$ where $s$ and $t$ are in the same component of $(V(G),A)$,
  that is,
  \begin{align*}
	Z_{st}(G;q,\boldgamma):= \sum_{\mathclap{\substack{A\subseteq E:\\s{\rm\ and\ }t{\rm\ in\ same\ component}}}}\ q^{\kappa(A)}
	\prod_{e\in A} \gamma_e. \end{align*}
  Similarly
  let $Z_{s|t}$ denote the contribution 
  to $\ZTutte{G}{q,\boldgamma}$ from configurations~$A$ in which
  $s$ and $t$ are in different components.

\subsection{Implementing new edge weights, series compositions, and parallel compositions }
  \label{sec:shiftdef}

  Our treatment of implementations, series compositions and parallel compositions
  is completely standard and is taken from \cite[Section 2.1]{GJplanar}. The reader
  who is familiar with this material can skip this section (which is included here for completeness).

  Fix $W \subseteq \algebraics$ and $q\in\algebraics$. 
  Let $w^*\in\algebraics$ be a weight   (which may not be in $W$) which we want to ``implement''.
  Suppose that 
  there is a graph~$\Upsilon$,
  with distinguished vertices $s$ and~$t$ 
  and a weight function $\hatbgamma: E(\Upsilon) \rightarrow W$
  such that
  \begin{equation}
  \label{eq:implement}
  w^* = q Z_{st}(\Upsilon;q,\hatbgamma)/Z_{s|t}(\Upsilon;q,\hatbgamma).
  \end{equation}
  In this case, we say that $\Upsilon$ and $\hatbgamma$ implement $w^*$
  (or even that $W$ implements~$w^*$).

  The purpose of ``implementing''  edge weights is this.
  Let $G$ be a graph with weight function $\bgamma$.
  Let $f$ be some edge of~$G$ with weight $\gamma_f=w^*$.
  Suppose that $W$ implements $w^*$.
  Let $\Upsilon$ be a graph with distinguished vertices $s$ and $t$
  with a weight function $\hatbgamma:E(\Upsilon)\rightarrow W$ 
  satisfying~(\ref{eq:implement}). 
  Construct the weighted graph~$G'$ 
  by replacing edge $f$ with a copy of~$\Upsilon$ (identify $s$ with either endpoint of~$f$
  (it doesn't matter which one) and identify $t$ with the other endpoint of $f$ and remove edge $f$).
  Let the weight function 
  $\bgamma'$ of $G'$ inherit weights from $\bgamma$ and 
  $\hatbgamma$ (so $\gamma'_e=\hat\gamma_e$ if $e\in E(\Upsilon)$ and
  $\gamma'_e = \gamma_e$ otherwise).
  Then the definition of the multivariate Tutte polynomial gives
  \begin{align}
	\label{eq:shift}
	\ZTutte{G'}{q,\bgamma'} & = \frac{Z_{s|t}(\Upsilon;q,\hatbgamma)}{q^2} 
	\ZTutte{G}{q,\bgamma}.
  \end{align}
  So, as long as $q\neq 0$ and $Z_{s|t}(\Upsilon;q,\hatbgamma)$ is easy to evaluate,
  evaluating the multivariate Tutte polynomial of $G'$ with weight function $\bgamma'$ is
  essentially the same as evaluating the multivariate Tutte polynomial of $G$ with weight function~$\bgamma$.
   
  Since the norm of the product of two complex numbers is the product of the norms,
  this reduces computing (or relatively approximating) the norm
  with weight function~$\bgamma$ to the problem of computing (or relatively approximating) the norm
  with weight function~$\bgamma'$.
  Also, since the argument of the product of two complex numbers is the
  sum of the arguments of the numbers,
  this reduces computing (or additively approximating) the argument
  with weight function~$\bgamma$ to the problem of computing (or additively approximating) the
  argument with weight function~$\gamma'$.

  Two especially useful implementations are series and parallel compositions.
  Parallel composition is the case in which $\Upsilon$ consists of two parallel edges $e_1$ and $e_2$
  with endpoints $s$ and $t$ and  $\hat\gamma_{e_1}=w_1$ and $\hat{\gamma}_{e_2}=w_2$.
  It is easily checked from Equation~(\ref{eq:implement})
  that $w^* = (1+w_1)(1+w_2)-1$. Also, the extra factor in Equation~(\ref{eq:shift}) cancels,
  so in this case $\ZTutte{G'}{q,\bgamma'} = \ZTutte{G}{q,\bgamma}$.

  Series composition is the case in which $\Upsilon$ is a length-2 path from $s$ to $t$ consisting of edges $e_1$ and $e_2$
  with $\hat\gamma_{e_1}=w_1$ and $\hat\gamma_{e_2}=w_2$.
  It is easily checked from Equation~(\ref{eq:implement})
  that $w^* =  w_1w_2/(q+w_1+w_2)$. Also, the extra factor in Equation~(\ref{eq:shift}) is $q+w_1+w_2$,
  so in this case 
  $\ZTutte{G'}{q,\bgamma'} = (q+w_1+w_2) \ZTutte{G}{q,\bgamma}$.
  It is helpful to note that
  $w^*$ satisfies
  $$\left(1+\frac{q}{w^*}\right) = \left(1+\frac{q}{w_1}\right) \left(1+\frac{q}{w_2}\right).$$

  We say that there is a ``shift'' 
  from $(q,\alpha)$ to $(q,\alpha')$ if 
  there is an implementation of $\alpha'$ consisting of some $\Upsilon$ and  
  $\hatw:E(\Upsilon)\rightarrow W$ where $W$ is the singleton set $W=\{\alpha\}$. 
  Taking $y=\alpha+1$ and $y'=\alpha'+1$
  and defining $x$ and $x'$ by $q=(x-1)(y-1)=(x'-1)(y'-1)$ we
  equivalently refer to this as a shift from $(x,y)$ to $(x',y')$.
  It is an easy, but important observation that shifts may be composed to 
  obtain new shifts.  So, if we have shifts from $(x,y)$ to $(x',y')$
  and from $(x',y')$ to $(x'',y'')$, then we also have a shift 
  from $(x,y)$ to $(x'',y'')$.

  The
  $k$-thickening of  \cite{JVW}
  is the parallel composition of $k$ edges of weight $\alpha$.
  It implements $\alpha'=(1+\alpha)^k-1$ and is a shift from $(x,y)$ to
  $(x',y')$ where $y' = y^k$ (and $x'$ is given by $(x'-1)(y'-1)=q$).
   Similarly, the $k$-stretch is the series composition of $k$ edges of weight $\alpha$.
  It implements an $\alpha'$
  satisfying
  $$1+\frac{q}{\alpha'}= {\left(1+\frac{q}{\alpha}\right)}^k,$$
  It is a shift from $(x,y)$ to 
  $(x',y')$ where $x'=x^k$.
  (In the classical bivariate $(x,y)$ parameterisation, there is effectively
  one edge weight, so the stretching or thickening is applied uniformly
  to every edge of the graph.)

  Thus, we have the following observation. 
  \begin{observation}  \label{obs:stretchthicken}
   The $k$-thickening operation
  gives the following polynomial-time reductions.
  \begin{itemize}
  \item  \CNIsing{\cst}{y^k } $\leq$ \CNIsing{\cst}{y },
  \item  \CAIsing{\addcst}{y^k} $\leq$  \CAIsing{\addcst}{y},  
  \item \SignReTutte{1+(x-1)(y-1)/(y^k-1),y^k} $\leq$ \SignReTutte{x,y}, where $y^k\neq 1 $, and
  \item \ZivIsing{y^k}  $\leq$ \ZivIsing{y}.
  \end{itemize}
  Similarly, $k$-stretching gives the following polynomial-time reductions for 
  $y\neq 1$. 
  \begin{itemize}
  \item
  \CNIsing{\cst}{1+2/((1+2/(y-1))^k-1) } $\leq$ \CNIsing{\cst}{y },
  \item
  \CAIsing{\addcst}{1+2/((1+2/(y-1))^k-1)} $\leq$ \CAIsing{\addcst}{y},   
  \item \SignReTutte{x^k,1+(x-1)(y-1)/(x^k-1)} $\leq$ \SignReTutte{x,y},
  \text{where $x^k\neq 1 $,} and
  \item \ZivIsing{1+2/((1+2/(y-1))^k-1)} $\leq$ \ZivIsing{y}. 
    \end{itemize}
  Similar statements hold for the relaxed versions of the problems.

  \end{observation}

\section{Hardness results for the Ising model}
\label{sec:hardIsing}

In this section we prove Theorems~\ref{thm:main} and \ref{thm:relaxed}.
 
\subsection{Real weights}

  First we gather some known results regarding approximating the partition function $\ZIsing{G}{\ybeta}$ of the Ising model when $\ybeta$ is an algebraic real number.

  If $\ybeta\in\{-1,0,1\}$, then 
  computing $\ZIsing{G}{\ybeta}$ is trivial from the definition~\eqref{eqn:Ising}.
  A classical result by Jerrum and Sinclair \cite{JS} settles the complexity of approximating $\ZIsing{G}{\ybeta}$ when $\ybeta>0$. They show that there
is a ``fully polynomial randomised approximation scheme'' (FPRAS) when $\ybeta>1$ and 
that it is\ \NP-hard to approximate the partition function when $0<\ybeta<1$.
  The negative case appears to be more complicated.
  Goldberg and Jerrum \cite{GJNPhard} showed that if $-1<\ybeta<0$, it is also\ \NP-hard to approximate $\ZIsing{G}{\ybeta}$,
  but if $\ybeta<-1$, the problem is equivalent to approximating the number of perfect matchings in a graph
  and it is not known whether there is an FPRAS. 
  Technically, neither Jerrum and Sinclair nor Goldberg and Jerrum
  worked over the algebraic numbers.
  In order to avoid issues of 
  real arithmetic, Jerrum and Sinclair used a computational model in which real arithmetic is performed with perfect
  accuracy, and Goldberg and Jerrum restricted attention to rationals.
  However, the operations in those papers are easily implemented over the algebraic real numbers.
  Using our notation, these results are summarised as follows.

  \begin{lemma} (\cite{JS,GJNPhard})
	 Suppose $\ybeta\in \algebraics$ and $\cst>1$.
	 Then
	\CNIsing{\cst}{\ybeta}
	\begin{itemize}
	  \item is in {\FP} if $\ybeta\in\{-1,0,1\} $;
      \item is in {\RP} if $\ybeta>1$; 
	  \item is NP-hard if $0<\ybeta<1$ or $-1<\ybeta<0$; and
	\item is equivalent in difficulty to approximately counting perfect matchings if $\ybeta<-1$.
	\end{itemize}
	\label{lem:norm:Ising:real}
  \end{lemma}
  
Technically, the  results in   \cite{JS, GJNPhard} were not
about the problem \CNIsing{\cst}{\ybeta} with fixed~$\cst$.
Instead, the accuracy parameter was viewed as part of the input
as in the following problem.
 \prob {\NIsing{\ybeta,\lambda}.}
{A (multi)graph $G$ and a positive integer $\lgeps$, in unary.}
{A  rational number $\apxnorm$
such that
\[\left(1-\tfrac{1}{\lgeps}\right) \apxnorm \leq  |\ZIsing{G}{\ybeta,\lambda}|
\leq \left(1+\tfrac{1}{\lgeps}\right) \apxnorm.\]}  
Nevertheless, the hardness results in Lemma~\ref{lem:norm:Ising:real} follow
easily from those papers using the following standard  powering lemma.
  
 \begin{lemma}
  \label{lem:duplicate}
  Let $\ybeta$ and $\lambda$ be algebraic numbers.
  For any $ \cst>1$, there are polynomial-time Turing reductions between \CNIsing{\cst}{\ybeta,\lambda}  
  and
  \NIsing{\ybeta,\lambda}.
\end{lemma}
\begin{proof}
The reduction from \CNIsing{\cst}{\ybeta,\lambda}  
  to
  \NIsing{\ybeta,\lambda}
is straightforward: 
Given an input~$G$ to \CNIsing{\cst}{\ybeta,\lambda},
 choose $\lgeps$   so
that $\cst \geq \lgeps/(\lgeps-1)$
and run an algorithm for \NIsing{\ybeta,\lambda} with inputs~$G$ and~$\lgeps$, returning the result.

The other direction is almost as easy.
Given an input $(G,\lgeps)$ to  \NIsing{\ybeta,\lambda},
choose an integer~$k$  sufficiently large (which does not depend on the size of $G$) so that
$(1-1/\lgeps)^k \leq 1/\cst$ and
$(1+1/\lgeps)^k \geq \cst$.
Then form $G_k$ by taking $k$ disjoint copies of $G$.
Run an algorithm for \CNIsing{\cst}{\ybeta,\lambda} with input $G_k$,
obtaining a number $\apxnorm$ such that 
$  {\apxnorm}/{\cst} \leq  |\ZIsing{G_k}{\ybeta,\lambda}| \leq  \cst \apxnorm$.
Then note that $\ZIsing{G_k}{\ybeta,\lambda} = \ZIsing{G}{\ybeta,\lambda}^k$,
so
$$ \left(1-\tfrac{1}{\lgeps}\right) {\apxnorm}^{1/k} \leq
{\apxnorm}^{1/k}/{\cst}^{1/k} \leq  |\ZIsing{G}{\ybeta,\lambda}| \leq  \cst^{1/k} \apxnorm^{1/k}
\leq {\apxnorm}^{1/k}\left(1+\tfrac{1}{\lgeps}\right),$$
so $\apxnorm^{1/k}$ is a suitable output.
 \end{proof}

  Note that the NP-hardness result for $ 0<\ybeta<1$ in Lemma~\ref{lem:norm:Ising:real} is
  essentially best possible in the sense that the problem is not much harder than NP. As \cite{GJNPhard} observed,
  the problem can be solved in randomised polynomial time using an oracle
  for an NP predicate by applying the bisection technique of Valiant and Vazirani~\cite{VV}.
  The situation is different for $\ybeta<0$.
  Goldberg and Jerrum \cite[Theorem 1, Region G]{GJSign} showed that it is \numP-hard to determine the sign of 
  $\ZIsing{G}{\ybeta}$ if $-1<\ybeta<0$.
  Again, they stated their theorem for the case in which $\ybeta$ is rational, but the proof applies
  equally well when $\ybeta$ is an algebraic real number.  
   In terms of our notation, they proved the following lemma. 
   
  \begin{lemma} (\cite{GJSign})
  For any algebraic real number $\ybeta\in (-1,0)$, 
  \SignReTutte{x,y} is \numP-hard, where 
  $x=1+2/(y-1)$.
  \end{lemma}

  If $\ybeta$ is real then $\ZIsing{G}{\ybeta}$ is real.
  Thus, either $\ZIsing{G}{\ybeta}=0$, or
  $\arg(\ZIsing{G}{\ybeta})\in \{0,\pi\}$. Hence, approximating
  the argument within $\pm \pi/3$ enables one to determine the sign of 
the real part.
Using the connection 
\eqref{eq:IsingTutte} between the Tutte polynomial and the
partition function of the Ising model  and Lemma~\ref{lem:zivapprox}
we immediately obtain the following corollary.

\begin{corollary} 
Suppose $\ybeta$ is an algebraic real number in the
range $\ybeta\in(-1,0)$. Then the problem
 \CAIsing{ (\pi/3)}{\ybeta} is \numP-hard and so is \ZivIsing{\ybeta}.
  \label{cor:nega:sign}
\end{corollary}

In fact, we can extend Goldberg and Jerrum's \numP-hardness interval-shrinking technique
from \cite{GJSign} to also obtain \numP-hardness
for the relaxed version of the problems.
We start with a general discussion of  interval shrinking.
Suppose that we have a linear function $f(\epsilon) = -\epsilon A  + B$ for positive~$A$ and~$B$
and that we wish to find a value $\hat{\epsilon}$ that is very close to the root $\epsilon^*=B/A$.
 Suppose that we also have an interval  $[\epsilon',\epsilon'']$
such that $f(\epsilon')>0$ and $f(\epsilon'')<0$.
Suppose that $\epsilon''-\epsilon' = \ell$ (so the interval has length~$\ell$).
Roughly, Goldberg and Jerrum  had at hand an oracle for
computing the sign of $f(\epsilon)$
(using an oracle for \SignReTutte{x,y})
 and, using this, it is easy to bisect
the interval, getting very close to $\epsilon^*$ by binary search.  

Using an oracle for the relaxed problem \NonzeroSignReTutte {x,y} 
we can compute the sign whenever it is positive or negative,
but we receive an unreliable answer for the sign of $f(\epsilon)$ if $f(\epsilon)=0$.
Nevertheless, we observe that having a reliable
answer in this case is not important for the progress of the
binary search.
If the binary search queries the value of~$f(\epsilon)$
and $f(\epsilon) \neq 0$
then the reply from the oracle is correct. Otherwise,
the bisection technique described above recurses into a sub-interval that
contains a zero of the function, as required.
Thus, we have the following lemma.
(We omit the formal proof since the lemma follows immediately from the observation that we have just made.)

\begin{lemma}\label{lem:nozeroadd} For any algebraic real number $\ybeta\in (-1,0)$, 
\NonzeroSignReTutte{x,y} is \numP-hard, where 
$x=1+2/(y-1)$.
Also, the problems \CNonzeroAIsing{(\pi/3)}{\ybeta} and \NonzeroZivIsing{\ybeta} are \numP-hard.
\end{lemma} 
 
We next show how to further extend the \numP-hardness interval-shrinking technique 
to obtain \numP-hardness for the problem  \CNonzeroNIsing{\cst}{\ybeta }. 
This requires new ideas, so we will provide more details.
Let us return to the discussion of interval shrinking.
Let $\eta=1/21$ (the exact value of $\eta$ is not important, but we fix it for concreteness).
Instead of having an oracle for the sign of $f(\epsilon) = - \epsilon A + B$,
we only will be able to assume that
we have an oracle that, on input $\epsilon$,
returns a value $\hat{f}(\epsilon)$
satisfying
$$(1-\eta) |f(\epsilon)| <
\tfrac{21}{22} |f(\epsilon)| \leq
\hat{f}(\epsilon)   
\leq \tfrac{22}{21} |f(\epsilon)|=
(1+\eta) |f(\epsilon)|,$$ 
except that again the value $\hat{f}(\epsilon)$ is completely unreliable if
$f(\epsilon)=0$.
Our strategy will be to divide the interval into
$10$ equal-length sub-intervals 
$[\epsilon_i,\epsilon_{i+1}]$ for $i\in\{0,\ldots,9\}$ with $\epsilon_0=\epsilon'$ and $\epsilon_{10}=\epsilon''$.
(The number $10$ is not chosen to be optimal --- however, it is easy to see that it suffices.
Changing the number of sub-intervals would influence the choice of~$\eta$ above.)
We then let $s_i$
be the sign (positive, negative, or zero) of $\hat{f}(\epsilon_i) - \hat{f}(\epsilon_{i+1})$, for each
$i \in \{0,\ldots,9 \}$. The $s_i$ values can be computed by the oracle.
Now recall that $\epsilon^*$ is the root $B/A$ of the function $f(\epsilon) = - \epsilon A + B$.
Consider next what happens if $\epsilon_i < \epsilon_{i+1} < \epsilon^*$
(so $f(\epsilon_i)> f(\epsilon_{i+1})>0$) .
In this case, 
\begin{align*}
\hat{f}(\epsilon_i) - \hat{f}(\epsilon_{i+1})
& \geq (1-\eta) f(\epsilon_i) - (1+\eta)f(\epsilon_{i+1})\\
&= A (\epsilon_{i+1}-\epsilon_{i} - \eta(2 \epsilon^* - \epsilon_i - \epsilon_{i+1}) ).
\end{align*}
Now $\epsilon_{i+1} - \epsilon_i \geq \ell/10$.
Also $\epsilon^*-\epsilon_i$ and $\epsilon^*-\epsilon_{i+1}$ are both at most $\ell$.
So since $\eta < 1/20$, $s_i$ is positive.
Similarly, if $\epsilon^* < \epsilon_i < \epsilon_{i+1}$ 
(so $f(\epsilon_{i+1} )< f(\epsilon_i)<0$  ) 
then 
\begin{align*}
\hat{f}(\epsilon_i) - \hat{f}(\epsilon_{i+1})
& \geq (1-\eta) (-f(\epsilon_i)) - (1+\eta)(-f(\epsilon_{i+1}))\\
&=  -A (\epsilon_{i+1}-\epsilon_{i} - \eta(2 \epsilon^* - \epsilon_i - \epsilon_{i+1}) ),
\end{align*}
so $s_i$ is negative.
If $\epsilon_i \leq \epsilon^*$ and $\epsilon_{i+1} \geq \epsilon^*$ then we
don't know what the value of $s_i$ will be. 
However, this is true for at most two values of~$i$.
So either
$s_0$, $s_1$, $s_2$ and $s_3$ are all positive
(in which case $\epsilon_2 < \epsilon^*$ and we can recurse on the interval 
$[\epsilon_2,\epsilon_{10}]$)
or $s_6$, $s_7$, $s_8$ and $s_9$ are all negative
(in which case $\epsilon_8 > \epsilon^*$ and we can recurse on the interval
$[\epsilon_0,\epsilon_8]$).
Either way, the interval shrinks to $4/5$ of its original length.

Applying this idea in the proof of \cite[Lemma 1]{GJSign} yields the following. 

\begin{lemma}
\label{lem:bisect} 
Suppose that $\gamma_1$ and $\gamma_2$ are algebraic reals with
$\gamma_1 \in (-2,-1)$ and $\gamma_2 \not\in [-2,0]$. Then
\CNonzeroNRCIsing{(\tfrac{22}{21})}{\gamma_1,\gamma_2 } is \numP-hard.
\end{lemma}
 
\begin{proof}
 Apart from the interval shrinking idea discussed above, 
 the proof is similar in structure to the proof of \cite[Lemma 1]{GJSign}. 
 We defer some calculations (which are unchanged) to \cite{GJSign} but we provide the
 rest of the proof to show how to get the stronger result. 
 We use the fact that the following problem is \numP-complete.
This was shown by
Provan and Ball~\cite{ProvanBall}. 
\prob{\MCCut.}
{A graph $G=(V,E)$ and distinguished vertices~$s,t\in V$.}
{$|\{S\subseteq E:\mbox{$S$ is a minimum cardinality $(s,t)$-cut in $G$}\}|$.}
 
We will give a Turing reduction from \MCCut\ 
to the problem \CNonzeroNRCIsing{(\tfrac{22}{21})}{\gamma_1,\gamma_2 }. 
 
Let $G,s,t$ be an instance of \MCCut.
Assume without loss of generality that $G$ has no edge from~$s$ to~$t$.
Let $n=|V(G)|$ and $m=|E(G)|$.
Assume without loss of generality that 
$G$ is connected and that 
$m\geq n$ is sufficiently large.
Let $k$ be the size of a minimum cardinality $(s,t)$-cut in~$G$ and
let $C$ be the number of size-$k$
$(s,t)$-cuts.

Let  $q=2$ and $M^* =  2^{4m}$.
Let $h$ be the smallest integer such that $(\gamma_2+1)^h-1>  M^*$
and let $M = (\gamma_2+1)^h-1$.
Note that we can implement~$M$ from~$\gamma_2$ via an $h$-thickening, and
$h$ is at most a polynomial 
in~$m$. 
 
Let $\delta =  4^{m}/M$. 
Let $\bM$ be the constant weight function which gives every
edge weight~$M$.
We will use the following facts:
\begin{equation} 
q M^m (1 - \delta)   \leq Z_{st}(G;q,\bM) \leq q M^m(1+\delta)
\label{thatitem}
\end{equation}
and
\begin{equation}
C M^{m-k}q^2(1-\delta) \leq Z_{s|t}(G;q,\bM) \leq C M^{m-k} q^2 (1+\delta).
\label{thisitem}
\end{equation}
Fact~(\ref{thatitem}) follows from
the fact that each of the (at most 
$2^m$) terms 
in $Z_{st}(G;q,\bM)$,  
other than the term with all edges in~$A$, has  size at most 
$M^{m-1} {q}^n$
and 
$ {2^m M^{m-1}{q}^n} \leq \delta
{M^m q}$.
Fact~(\ref{thisitem}) follows from
the fact that all terms in $Z_{s|t}(G;q,\bM)$ are complements of $(s,t)$-cuts.
If  more than $k$ edges are cut  then the term is at most
$M^{m-k-1} q^n$
and 
$$2^m
M^{m-k-1} 
q^n \leq \delta C M^{m-k}q^2.$$
 
For a parameter~$\epsilon$ in the open interval $(0,1)$ which we will tune later, let
$\gamma' = -1-\epsilon \in (-2,-1)$. 
We will discuss the implementation of $\gamma'$  later.
Let $G'$ be the  graph formed from $G$ by adding an edge from $s$ to~$t$.
Let $\bgamma$ be the edge-weight function for~$G'$ that assigns weight~$M$ to every edge of~$G$ and
assigns weight $\gamma'$ to the new edge.
Using the definition of the (random cluster) Tutte polynomial,
Goldberg and Jerrum noted that
\begin{align}
\nonumber
\ZTutte{G'}{2,\bgamma} &= Z_{st}(G;2,\bM)(1 + \gamma') + Z_{s|t}(G;2,\bM) \left(1 + \frac{\gamma'}{2}\right) \\
&= - \epsilon Z_{st}(G;2,\bM)  + Z_{s|t}(G;2,\bM) \left(1 - \frac{ 1+\epsilon}{2}\right).
\label{eq:Z}
\end{align}
It is easily checked  that
$\ZTutte{G'}{2,\bgamma} $
is positive if $\epsilon$ is sufficiently small ($\epsilon=M^{-2m}$ will do) and it is negative at $\epsilon=1$. 
Thus, viewing $\ZTutte{G'}{2,\bgamma} $
as a function of $\epsilon$, we can
perform interval shrinking (as discussed before the statement of the lemma)
to find a value of $\epsilon$ for which $\ZTutte{G'}{2,\bgamma}$ is very close to~$0$.
The interval shrinking uses an oracle for \CNonzeroNRCIsing{(\tfrac{22}{21})}{\gamma_1,\gamma_2 }.

If we find an $\epsilon$ where $\ZTutte{G'}{q,\bgamma}=0$, then
for this value of~$\epsilon$, we have
$\epsilon Z_{st}(G;q,\bM)
= Z_{s|t}(G;q,\bM) \left(1 - \frac{ 1+\epsilon}{2}\right)$.
Thus, using~$\epsilon$, we can calculate the fraction $Z_{s|t}(G;q,\bM)/Z_{st}(G;q,\bM)$.
Plugging this (known) value into   (\ref{thatitem}) and (\ref{thisitem}),
we obtain
$$\frac{C q (1-\delta)}{M^k(1+\delta)} \leq \frac{Z_{s|t}(G;q,\bM)}{Z_{st}(G;q,\bM)}
\leq \frac{C q(1+\delta)}{M^k(1-\delta)}.$$
Now, we don't know~$k$, 
but $C$ is an integer between~$1$ and $2^m$, 
whereas $M>2^{4m}$, so  there is only one value of~$k$ that gives
a solution~$C$ in the right range. Using the value of~$k$, we can calculate~$C$
exactly.

Technical issues arise both because
we are somewhat constrained in what values~$\epsilon$ we can implement
and because we won't be able to
discover the exact value of~$\epsilon$ that we need
(but we will be able to approximate it closely).
These technical issues provide no more difficulty than they did in~\cite{GJSign}.
Suppose first that we are able, for any given 
$\epsilon \in (M^{-2m},1)$ to implement $\gamma'=-1-\epsilon$.
 Then our basic strategy is to do the interval shrinking, repeatedly sub-dividing the
current interval $\Theta(\log(M^{m^2}))$ 
times, so eventually we'll get an interval of width at most $M^{-m^2}$
that contains an $\epsilon$ where $\ZTutte{G'}{2,\bgamma}=0$.
Goldberg and Jerrum~\cite{GJSign} have already shown
that knowing such an interval enables the exact calculation of~$C$
(so having a small interval is OK --- it is not necessary to know~$\epsilon$ exactly).

The only issue, then, is implementing the weights $\gamma'=-1-\epsilon$ during the  interval shrinking.
As in \cite{GJSign} we cannot expect to implement any particular desired $\gamma'$ precisely.
However, using stretching and thickening, we can implement a value that
is within an additive error of $M^{-m^2}/20$ of any desired $\epsilon$,
and this suffices. The fact that we have algebraic, rather than rational, numbers
is irrelevant since stretchings and thickenings can be computed on algebraic numbers.
\end{proof}

Using stretching and thickening, we get the following corollary.

\begin{corollary}
\label{cor:used}
Suppose $\cst>1$ and that
 $y\in(-1,0)$ is  an algebraic real number.  Then  
\CNonzeroNIsing{ \cst}{y } is \numP-hard.
\end{corollary}
\begin{proof}
We first show that \CNonzeroNIsing{(22/21)}{y } is \numP-hard.
Consider the edge interaction $y\in(-1,0)$.
Using the correspondence from~\eqref{eq:IsingTutte} and
\eqref{eq:TutteRC}, this corresponds directly to the quantity
$\gamma_1 \in (-2,-1)$ in Lemma~\ref{lem:bisect}.
We now consider how to use~$y$ to implement the quantity~$\gamma_2$.
A $2$-thickening from $(x,y)$ gives an effective weight 
$(x',y')$ with $y'=y^2\in(0,1)$ and $x'=2/(y'-1)+1<-1$.
Then a $2$-stretch from $(x',y')$ gives an effective weight
$(x'',y'')$ with $x'' = (x')^2 >1$ and $y'' = 2/(x''-1)+1 >1$, corresponding to $\gamma_2 > 0$, as required.

The reduction from \CNonzeroNIsing{(22/21)}{y }
to \CNonzeroNIsing{ \cst}{y } follows from Lemma~\ref{lem:duplicate}.
\end{proof}

Using Lemma~\ref{lem:zivapprox}  and the trivial reduction from
\CNonzeroNIsing{\cst}{y } to \CNIsing{\cst}{y }
and from \NonzeroZivIsing{y} to \ZivIsing{y} we get the following.

\begin{corollary}
 \label{cor:realnegy}
 Let $\ybeta\in(-1,0)$ be an algebraic real number.  
 Then  for any $\cst>1$,  \CNIsing{ \cst}{\ybeta } and \NonzeroZivIsing{\ybeta} and \ZivIsing{\ybeta} are \numP-hard.
\end{corollary}

\subsection{Complex weights}
 
\begin{lemma}
  Let $\theta\in[0,2\pi)$ and $\theta\not\in\{0,\frac{\pi}{2},\pi,\frac{3\pi}{2}\}$.
  There  is a positive integer~$k$ and an integer~$l$ such that $k\theta+2\pi l\in(\frac{\pi}{2},\pi)\cup(\pi,\frac{3\pi}{2})$.
  \label{lem:angle}
\end{lemma}
\begin{proof}
  Clearly if $\theta\in(\frac{\pi}{2},\pi)\cup(\pi,\frac{3\pi}{2})$ then we are done by letting $k=1$ and $l=0$.
  Otherwise $\theta\in(0,\frac{\pi}{2})\cup(\frac{3\pi}{2},2\pi)$.
  If $\theta$ is an irrational fraction of $2\pi$ then we can go through the whole unit circle by taking multiple of $\theta$.
  So assume $\theta=\frac{2\pi \pp}{\qq}$ where $\pp$ and $\qq$ are co-prime and $\qq=3$ or $\qq\geq 5$ as $\theta\not\in\{0,\frac{\pi}{2},\pi,\frac{3\pi}{2}\}$.
  Moreover $\qq=3$ contradicts $\theta\in(0,\frac{\pi}{2})\cup(\frac{3\pi}{2},2\pi)$.
  Hence $\qq\geq5$ and there exists an integer $t\neq \qq/2$ such that $\qq<4t<3\qq$.
  As $\pp$ and $\qq$ are relatively prime, there exist integers $l_1,l_2$ such that $l_1\pp+l_2\qq=1$ and $l_1>0$.
  It is easy to see that $tl_1\theta = \frac{2\pi tl_1\pp}{\qq}=-2\pi tl_2+\frac{2\pi t}{\qq}$.
  As $t/\qq\in(1/4,1/2)\cup(1/2,3/4)$ we have that $\frac{2\pi t}{\qq}\in(\frac{\pi}{2},\pi)\cup(\pi,\frac{3\pi}{2})$.
  The lemma follows by taking $k = t l_1$ and $l = t l_2$.
\end{proof}

The following lemma enables us to determine the complexity of 
evaluating the Ising partition function when the complex edge interaction 
$y\in\algebraics$ is on the unit circle.
 
\begin{lemma}
  \label{lem:norm1}
  Let $y=e^{\I\theta}\in\CC$ be an algebraic complex number such that $\theta\in[0,2\pi)$ and $\theta\not\in\{0,\frac{\pi}{2},\pi,\frac{3\pi}{2}\}$.
  There exists an algebraic real number $y'\in(-1,0)$ 
  that can be implemented by a sequence of stretchings and thickenings from~$y$.
\end{lemma}
\begin{proof} 

  By Lemma~\ref{lem:angle}, 
  there  is a positive integer~$k$ and an integer~$l$ such that $k\theta+2\pi l\in(\frac{\pi}{2},\pi)\cup(\pi,\frac{3\pi}{2})$.  
  As a $k$-thickening realizes $y^k=e^{\I k\theta}$,
  we may assume $\theta\in(\frac{\pi}{2},\pi)\cup(\pi,\frac{3\pi}{2})$.

  Since $\theta\not\in\{0,\frac{\pi}{2},\pi,\frac{3\pi}{2}\}$, we have $\cos\theta\neq 1$ and $\sin\theta\cos\theta\neq0$.
  The latter implies that $\sin\theta+\cos\theta\neq 1$.
  Let $x=\frac{y+1}{y-1}$.
   Note that $x=\frac{\sin\theta}{\cos\theta-1}\I$.
  Moreover $\theta\in(\frac{\pi}{2},\pi)\cup(\pi,\frac{3\pi}{2})$,
  implies that $\cos\theta<0$ and hence $|x|<1$.
  We do a $2$-stretch and the effective weight is $y'=1-\frac{2}{|x|^2+1}\in(-1,0)$.
\end{proof}
 
Combining Lemma~\ref{lem:norm1} with  Observation~\ref{obs:stretchthicken},
Corollary~\ref{cor:used},
Lemma~\ref{lem:nozeroadd} and Corollary~\ref{cor:realnegy} 
we get the following corollary,
which applies to
the problems  \CNonzeroNIsing{\cst}{\ybeta },
\CNonzeroAIsing{(\pi/3)}{\ybeta} and \NonzeroZivIsing{\ybeta} 
and also to the unrelaxed versions
\CNIsing{\cst}{\ybeta },
\CAIsing{(\pi/3)}{\ybeta} and \ZivIsing{\ybeta}. 
 
\begin{corollary}
  \label{cor:norm1:hard} 
  Let $\ybeta=e^{\I\theta}\in\CC$ be an algebraic complex number such that $\theta\in[0,2\pi)$ and $\theta\not\in\{0,\frac{\pi}{2},\pi,\frac{3\pi}{2}\}$. 
  Then for any $ \cst>1$, \CNonzeroNIsing{\cst}{\ybeta }, \CNonzeroAIsing{(\pi/3)}{\ybeta} and \NonzeroZivIsing{\ybeta} are \numP-hard. 
  Hence, so are the un-relaxed versions of all three problems. 
\end{corollary}

The hardness on the unit circle extends directly to the whole imaginary axis. 

\begin{lemma}    \label{lem:imaginary}
  Suppose $y=r\I$ and $r\neq 0,\pm1$ where $r$ is algebraic.
   There exists an algebraic real number $y'\in(-1,0)$ 
 that can be implemented by a sequence of stretchings and thickenings from~$y$.    
\end{lemma}

\begin{proof}
  If $0<|y|<1$, then a $2$-thickening yields effective weight $y^2=-r^2\in(-1,0)$.
  Let $y'=-r^2$ and the claim holds.

  Otherwise suppose $|y|>1$.
  We know that a $k$-stretch yields the weight
  $z_k = 1+2/(x^k-1)$ where 
  $x = 1+2/(y-1)=(y+1)/(y-1)$.
  Re-arranging, we find that $z_k=\frac{(y+1)^k+(y-1)^k}{(y+1)^k-(y-1)^k}$.
  We will now argue that $z_k$ is purely imaginary.  
  To see this, note that monomials in the numerator all have degrees of the same parity as $k$, 
  whereas those in the denominator have degrees of the same parity as $k-1$.
  Therefore, it must be the case 
  that the numerator is real and the denominator is purely imaginary, or vice versa.
  In either case,   $z_k$ is purely imaginary.
  Therefore, if we can find a positive integer $k$ such that $0<|z_k|<1$ then 
  we have reduced our problem  to the previous case.
 
  Since $y$ is purely imaginary, we have that $|y+1|=|y-1|$.
  Since $x=(y+1)/(y-1)$, this implies that $|x|=1$.
  It is easy to see that $0<|z_k|<1$ if and only if $|x^k+1|<|x^k-1|$ and $x^k\neq-1$.
  This in turn is equivalent to $\arg\left( x^k \right)\in\left( \frac{\pi}{2},\pi\right)\cup\left(\pi,\frac{3\pi}{2} \right)$.
  By Lemma~\ref{lem:angle}, such a $k$ always exists unless $\arg(x)=\frac{t\pi}{2}$ where $t=0,1,2,3$.
  In these cases $y=\pm 1, \pm \I$, which contradicts our assumption.
\end{proof}

Combining Lemma~\ref{lem:imaginary} with Observation~\ref{obs:stretchthicken},
Corollary~\ref{cor:used},
Lemma~\ref{lem:nozeroadd} and Corollary~\ref{cor:realnegy},
we get the following corollary.
 
\begin{corollary}\label{cor:twobad} 
 Let $y=r\I$ where $r\neq 0,\pm1$  and $r$ is algebraic. Let $\cst>1$.
 Then \CNonzeroNIsing{\cst}{\ybeta },
\CNonzeroAIsing{(\pi/3)}{\ybeta} and \NonzeroZivIsing{\ybeta} are \numP-hard. Hence, so are the
un-relaxed versions of all three problems.
 \end{corollary}

Finally, this hardness can be extended to some algebraic complex numbers
off of the unit circle.

\begin{lemma}
  Let $y=re^{\I\theta}$ be an algebraic complex number such that $r>0$ and $\theta=\frac{\pp\pi}{2\qq}$, 
  where $\pp$ and $\qq$ are two co-prime positive integers and $\pp$ is odd.
   There exists an algebraic real number $y'\in(-1,0)$ 
 that can be implemented by a sequence of stretchings and thickenings from~$y$.    
  \label{lem:rational:even}
\end{lemma}
\begin{proof}
  If $r=1$ then we are done by Lemma \ref{lem:norm1}.
  Otherwise $r\neq 1$ and by a $\qq$-thickening it reduces to the case of Lemma \ref{lem:imaginary}.
\end{proof}

\begin{corollary}
  \label{cor:onebad} 
  Let $y=re^{\I\theta}$ be an algebraic complex number such that $r>0$ and $\theta=\frac{\pp\pi}{2\qq}$, 
  where $\pp$ and $\qq$ are two co-prime positive integers and $\pp$ is odd.
  Then for any $ \cst >1$, \CNonzeroNIsing{\cst}{\ybeta }, \CNonzeroAIsing{(\pi/3)}{\ybeta} and \NonzeroZivIsing{\ybeta} are \numP-hard. 
  Hence, so are the un-relaxed versions of all three problems. 
\end{corollary}

To obtain obtain \NP-hardness results
for other values of $y$, we start with the  well-known \NP-hard problem \maxcut.

\prob{  \maxcut.}{ A (multi)graph $G$ and a positive integer $b$.}
{  Is there a cut of size at least $b$.}

\begin{lemma} Suppose $\cst>1$.
  Let $y$ be an algebraic complex number such that $|y|<1$ and $y\neq 0$.
  Then  
  \CNonzeroNIsing{\cst}{y} is \NP-hard and so is \NonzeroZivIsing{y}.
  \label{lem:norm<1:nphard}
\end{lemma}

\begin{proof}
  We will reduce \maxcut\ to \CNonzeroNIsing{\cst}{y}.
  Given a graph $G$ and a constant $b$, 
  we want to decide whether $G$ has a cut of size at least $b$.
  We do a $k$-thickening on $G$, where $k$ is the least positive integer such that $2^m|y|^k<1/4$.
  Then the effective edge weight is $y_k=y^k$.
  Clearly $|y_k|=|y|^k<1$.
  
  Suppose the maximum cut of $G$ has size~$c$.
  Now rewrite \eqref{eqn:Ising} as
  \begin{align*}
	\ZIsing{G}{y_k}=\sum_{i=0}^{c} C_i y_k^{m-i},
  \end{align*}
  where $m$ is the number of edges in $G$ and $C_i$ is the number of configurations under which there are exactly $i$~bichromatic edges.
  Since the maximum cut of~$G$ has size~$c$ and $G$ has $m$~edges,  $\sum_{i=0}^{m-c} C_i=2^m$.
  Also, since $2^m |y_k| <1$, the $i=c$ term dominates the sum,
  so \ZIsing{G}{y_k}  is not equal to~$0$.
  
  If $c\geq b$, then our choice of $k$ together with the triangle inequality implies that
  \begin{align*}
	|\ZIsing{G}{y_k}|&\ =\ |C_c y_k^{m-c}+\sum_{i=0}^{c-1}C_i y_k^{m-i}|\ >\ C_c|y_k|^{m-c}-2^m|y_k|^{m-c+1}\\
	&\ >\ |y_k|^{m-c}|1-2^m|y|^k|\ >\ \tfrac34 |y_k|^{m-b}.
  \end{align*}
  Otherwise we have $c\leq b-1$ and
  \begin{align*}
	|\ZIsing{G}{y_k}|&\ =\ |\sum_{i=0}^{c} C_i y_k^{m-i}|\ <\ \sum_{i=0}^{c} C_i |y_k|^{m-i}\\
	&\ \leq\ 2^m |y_k|^{m-b+1}\ <\ \tfrac14 |y_k|^{m-b}
  \end{align*}
  again by the triangle inequality and $2^m|y_k|<1/4$.
  Therefore we could solve \maxcut\ in polynomial time using an oracle
  for \CNonzeroNIsing{1.1}{y_k}. By
  Observation~\ref{obs:stretchthicken} 
  it suffices to use an oracle for
  \CNonzeroNIsing{1.1}{y}. 
  By Lemma~\ref{lem:duplicate}, 
  an oracle for   \CNonzeroNIsing{\cst}{y} will do. Finally, 
    Lemma~\ref{lem:zivapprox} gives the result for \NonzeroZivIsing{y}.
\end{proof}

The other case, when the norm of~$y$ is larger than $1$, can be 
shown to be NP-hard by reduction from
  the previous case, unless the edge weight is real.

\begin{lemma} Suppose $\cst>1$.
 Let $y$ be an algebraic complex number such that   $|y|>1$ and $y\not\in\RR$.
  Then 
  \CNonzeroNIsing{\cst}{y} is NP-hard and so is \NonzeroZivIsing{y}.
  \label{lem:norm>1:nphard}
\end{lemma}
\begin{proof}
  We will prove that there exists a positive integer $k$ such that the effective weight $y_k$ of a $k$-stretch satisfies $|y_k|<1$.
  Then we are done by Lemma~\ref{lem:norm<1:nphard}.

  Recall that $y_k=\frac{x^k+1}{x^k-1}$ where $x=\frac{y+1}{y-1}$.
  Clearly $|y_k|<1$ if and only if $|x^k+1|<|x^k-1|$.
  The latter is equivalent to $\arg(x^k)=k\arg(x)\in(\pi/2,3\pi/2)$
  (plus some integer multiple of $2\pi$).
  Let $\theta=\arg(x)\in[0,2\pi)$.
  The fact that $|y|>1$  implies that $\theta\in[0,\pi/2)\cup(3\pi/2,2\pi)$.
  If $\theta=0$, then $y\in\RR$, which is a contradiction.
  Therefore $\theta\in(0,\pi/2)\cup(3\pi/2,2\pi)$.
  By Lemma~\ref{lem:angle}, there is a positive integer~$k$ and 
  and integer~$l$ such that $k\theta+2\pi l\in(\pi/2,\pi)\cup(\pi,3\pi/2)\subset(\pi/2,3\pi/2)$.
  This is exactly what we need.
  Moreover, $k$ does not depend on the input $G$.
  This finishes our proof.
\end{proof}
 
\subsection{Proof of Theorems~\ref{thm:main} and \ref{thm:relaxed}}

Theorems~\ref{thm:main} and~\ref{thm:relaxed} follow
from the following  combined theorem. The hardness result in Item~\ref{realthree} of Theorem~\ref{thm:main}
(and its counterpart in Theorem~\ref{thm:relaxed})
follows from Item~\ref{cNPitem} of the  combined theorem.

\begin{theorem} 
Let $y=r e^{i  \theta}$ be an algebraic complex number  with $\theta\in [0,2\pi)$.
Suppose $\cst>1$.
\begin{enumerate}
\item \label{citemexact} If $y=0$ or if
$r=1$ and $\theta \in \{0,\frac{\pi}{2},\pi,\frac{3\pi}{2}\}$
then \CNIsing{\cst}{ y }, \CAIsing{(\pi/3)}{ y} and \ZivIsing{y} 
are in \FP.
  
\item \label{crealtwo} If $y>1$ is a real number then \CNIsing{\cst}{ y} 
and \ZivIsing{y} are
in \RP\ 
and  \CAIsing{(\pi/3)}{ y} is in \FP.

\item \label{crealthree} If $y$ is a real number in $(0,1)$ then   \CAIsing{(\pi/3)}{ y} is in \FP.

\item \label{crealfour} If $y<-1$ is a real number then \CNonzeroNIsing{\cst}{ y} is equivalent in complexity to
the problem of approximately counting perfect matchings in graphs 
and \NonzeroZivIsing{y} is as hard.
However, 
\CAIsing{(\pi/3)}{ y} is in \FP. 

\item \label {crealfive} If $y$ is a real number in $(-1,0)$ then \CNonzeroNIsing{ \cst}{ y },
\CNonzeroAIsing{ (\pi/3)}{y}   
and \NonzeroZivIsing{y}
are $\numP$-hard.

\item \label {ccircleone} If $r=1$ and $\theta\not\in\{0,\frac{\pi}{2},\pi,\frac{3\pi}{2}\}$
then \CNonzeroNIsing{\cst}{ y }, \CNonzeroAIsing{(\pi/3)}{ y} and \NonzeroZivIsing{ y} are \numP-hard. 
\item \label{cimaginaries} If $\theta \in \{ \frac{\pi}{2}, \frac{3\pi}{2}\}$
and $r\not\in\{-1,0,1\}$ then \CNonzeroNIsing{\cst}{ y},
\CNonzeroAIsing{(\pi/3)}{ y} and \NonzeroZivIsing{ y} are \numP-hard. 
\item \label{ccoprime} If $r>0$ and $\theta=\frac{\pp\pi}{2\qq}$, 
  where $\pp$ and $\qq$ are two co-prime positive integers and $\pp$ is odd then 
  \CNonzeroNIsing{\cst}{ y }, \CNonzeroAIsing{(\pi/3)}{ y} and \NonzeroZivIsing{ y} are \numP-hard.   
\item \label{cNPitem}
If $r<1$ and $y\neq 0$ then    \CNonzeroNIsing{\cst}{ y}  and  \NonzeroZivIsing{ y} are \NP-hard.
 \item \label{cNPlastitem}
 If $r>1$ and $\theta\not\in\{0,\pi\}$ then  \CNonzeroNIsing{\cst}{ y}  and  \NonzeroZivIsing{ y} are \NP-hard.
   \end{enumerate}

\end{theorem}
\begin{proof}
Item~\ref{citemexact} is from \cite{JVW}. 
The randomised algorithm for \CNIsing{\cst}{ y} referred to in Item~\ref{crealtwo} is from~\cite{JS}. 
See also Lemma~\ref{lem:norm:Ising:real} and the surrounding
text for a discussion of algebraic numbers and accuracy parameters.
The same algorithm can be used for \ZivIsing{y} because 
$\ZIsing{G}{\ybeta}$  is real and positive
so an approximation $\apxnorm$ satisfing
\[\left(1-\tfrac{1}{\lgeps}\right) \apxnorm \leq  \ZIsing{G}{\ybeta,\lambda}
\leq \left(1+\tfrac{1}{\lgeps}\right) \apxnorm\]
also satisfies 
 $d(\apxnorm, \ZIsing{G}{\ybeta,\lambda}) \leq  \tfrac{1}{\lgeps}$. 
The deterministic algorithm referred to in Items~\ref{crealtwo} and~\ref{crealthree}
is trivial because the argument of a positive real number is~$0$. 
The approximation equivalence in Item \ref{crealfour} is from \cite{GJNPhard},
since one can decide in polynomial time the existence of perfect matchings to lift the non-zero restriction.
The hardness for \NonzeroZivIsing{y} follows from Lemma~\ref{lem:zivapprox}.
The deterministic sign algorithm in Item \ref{crealfour} is from \cite{GJSign}.
Item~\ref{crealfive} is from    
Lemma~\ref{lem:nozeroadd}  and Corollary~\ref{cor:used}
and Lemma~\ref{lem:zivapprox}.
Item~\ref{ccircleone} is from Corollary~\ref{cor:norm1:hard}.
Item~\ref{cimaginaries} is from Corollary~\ref{cor:twobad}.
Item~\ref{ccoprime} is from Corollary~\ref{cor:onebad}.
Item~\ref{cNPitem} is from Lemma~\ref{lem:norm<1:nphard}.
Finally,
item~\ref{cNPlastitem} is from Lemma~\ref{lem:norm>1:nphard}.
\end{proof}  
 
\section{Quantum circuits and counting complexity}
In this section we explain the connection between quantum computation and complex weighted Ising models.
We begin with some basic notions about quantum circuits.
We view qubits $\ket{0}$ and $\ket{1}$ as column vectors $\BinVec{1}{0}$ and $\BinVec{0}{1}$.
Similarly $\bra{0}$ and $\bra{1}$ are row vectors $(1,0)$ and $(0,1)$.
For $\mathbf{x}\in\{0,1\}^n$, let $\ket{\mathbf{x}}$ denote the tensor product $\otimes_{\ii=1}^{n}\ket{x_\ii}$ and $\bra{\mathbf{x}}$ is similar.

Suppose $C$ is a quantum circuit on $n$ qubits and consists of $m$ quantum gates $U_1,\dots,U_m$ sequentially.
A quantum gate is a function taking $k$ input and $k$ output variables and returning a value in $\mathbb{C}$.
Such a gate is called $k$-local and has a natural $2^k$ by $2^k$ square unitary matrix representation.
In a circuit we also need to specify on which qubits the gate acts upon.
To make the notation uniform we view unaffected qubits as simply copied 
and associate each quantum gate with the following $2^n$ by $2^n$ square unitary matrix.
Let $U$ be a quantum gate and $\mathbf{x,y}\in\{0,1\}^n$ two vectors specifying the input and output on all $n$ qubits.
Define the $2^n$ by $2^n$ matrix $M_U$ corresponding to gate $U$ as $M_{U;\mathbf{x},\mathbf{y}}=U(\mathbf{x},\mathbf{y})$.

For example, let $H$ be the Hadamard gate 
$\tfrac{1}{\strut\sqrt{2}}\BinMatrix{1}{1}{1}{-1}$ 
acting on the first qubit and suppose there are two qubits in total, illustrated as in Figure \ref{fig:h1}.
Then the matrix $M_H$ is 
$\tfrac{1}{\strut\sqrt{2}}\BinMatrix{1}{1}{1}{-1}\otimes \BinMatrix{1}{0}{0}{1}=
\frac{1}{\strut\sqrt{2}}
\left[
  \begin{smallmatrix}
	1 & 0 & 1 & 0 \\
	0 & 1 & 0 & 1 \\
	1 & 0 & -1 & 0 \\
	0 & 1 & 0 & -1 \\	
  \end{smallmatrix}
\right]$.
 
Using this notation, given an input $\mathbf{x}\in\{0,1\}^n$,
the output of the quantum circuit $C$ is a random variable $\mathbf{Y}$ subject to the distribution
\begin{align}
  \label{eqn:quantum:marginal}  
  \Pr_C(\mathbf{Y}=\mathbf{y})=\left|\bra{\mathbf{y}}\prod_{\ii=1}^{m}M_{U_{m+1-\ii}}\ket{\mathbf{x}}\right|^2,
\end{align}
where $\mathbf{y}\in\{0,1\}^n$.
It is not necessary that we measure all qubits in the output.
We may measure a subset $I$ of all $n$ qubits.
Let $\mathbf{y'}\in\{0,1\}^\ss$ where $|I|=\ss$.
Then the output is a random variable $\mathbf{Y'}$ subject to the distribution
\begin{align}
  \label{eqn:quantum:marginal:partial}  
  \Pr_{C;I}(\mathbf{Y'}=\mathbf{y'})=\sum_{\mathbf{z}\in\{0,1\}^{n} \rm{\ such\ that\ }\mathbf{z}|_{I}=\mathbf{y'}}\Pr_{C}(\mathbf{Y}=\mathbf{z}).
\end{align}

Alternatively, we may treat such marginal probability in the counting perspective, as a partition function in the ``sum of product'' fashion.
First let us consider composing two quantum gates, say $U_1$ and $U_2$.
Let the input variables of~$U_1$ be $x_1,\ldots,x_n$.
Let $z_1,\ldots,z_n$ be the variables on the wires between~$U_1$ and~$U_2$.
Finally, let $y_1,\ldots,y_n$ be the outputs of~$U_2$.
We use $\sigma(\mathbf{x})$ to denote an assignment of values in $\{0,1\}$ to the variables  $x_1,\ldots,x_n$.
We use $\sigma(\mathbf{y})$ and $\sigma(\mathbf{z})$ similarly.
Then the composition $U$ of $U_1$ followed by $U_2$ is given by 
\begin{align}
  U(\mathbf{x},\mathbf{y})=\sum_{\sigma(\mathbf{z})} U_1(\mathbf{x},\sigma(\mathbf{z}))U_2(\sigma(\mathbf{z}),\mathbf{y}).
  \label{eqn:gate:multiplication}
\end{align}
Figure \ref{fig:u1u2} illustrates the composition of gate $U_1$ acting upon qubits $2,3,4$ followed by $U_2$ acting upon $1,2$.
In the matrix notation, it is easy to see that $M_{U}=M_{U_1}M_{U_2}$.

\begin{figure}[t]
  \centering
  \begin{minipage}{.35\textwidth}
    \begin{center}
      \includegraphics{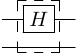}
    \end{center}
	\caption{Gate $H$ applying only on the first qubit.}
    \label{fig:h1}
  \end{minipage}%
  \hspace{0.03\textwidth}
  \begin{minipage}{.6\textwidth}
    \begin{center}
      \includegraphics{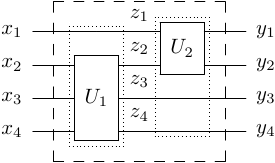}
    \end{center}
    \caption{Two quantum gates $U_1$ and $U_2$ composed together.}
    \label{fig:u1u2}
  \end{minipage}
\end{figure}

We now associate an intermediate variable $z_{\ii,\jj}$ to each edge   on qubit $\jj$ between gate $U_{\ii}$ and $U_{\ii+1}$ for all $2\leq \ii \leq m-1$ and $1\leq \jj\leq n$.
Denote by $\mathbf{z_\ii}$ the vector $\{z_{\ii,\jj}\mid 1\leq \jj\leq n\}$ and $\mathbf{z}=\cup_{\ii=2}^{m-1} \mathbf{z_\ii}$.
As the initial input and output of a quantum circuit are column vectors and row vectors respectively,
they may be treated as function/gates with no output variables or no input variables.
In particular, on the product input state $\ket{\mathbf{x}}$ input variables are set to $\{x_\jj\}$ where $\mathbf{x}\in\{0,1\}^n$.
Using \eqref{eqn:gate:multiplication} recursively we can rewrite \eqref{eqn:quantum:marginal} as follows:
\begin{align}
  \Pr_C(\mathbf{Y}=\mathbf{y})=\left|
  \sum_{\sigma:\mathbf{z}\rightarrow\{0,1\} }
  U_1(\mathbf{x},\sigma(\mathbf{z_1}))
  U_m(\sigma(\mathbf{z_{m-1}}),\mathbf{y})
  \prod_{\ii=2}^{m-1} U_\ii(\sigma(\mathbf{z_{\ii-1}}),\sigma(\mathbf{z_\ii}))
  \right|^2.
  \label{eqn:marginal:partitionfunction}
\end{align}

To simulate classically a quantum circuit, one can either 
(approximately) compute the probability $\Pr_C(\mathbf{Y}=\mathbf{y})$ 
--- this is called ``strong simulation'' ---
or one can sample from a distribution that is sufficiently close
to the one   given by \eqref{eqn:quantum:marginal} or \eqref{eqn:marginal:partitionfunction}.
This is called ``weak simulation''

\subsection{\IQP\ and the Ising partition function}
\label{sec:IQP}

\IQP, which stands for ``instantaneous quantum polynomial time'', is 
characterised by a restricted class of quantum circuits introduced by Shepherd and Bremner \cite{SB09}.
Bremner et al.~\cite{BJS11} showed
  that if \IQP\ can be simulated classically in the sense of ``weak simulation'' with multiplicative error,
then the polynomial hierarchy collapses to the third level.
 Fujii and Morimae \cite{FM} showed that the marginal probabilities of possible outcomes
of \IQP\ circuits correspond to partition functions of Ising models with complex edge weights.

The key property of \IQP\ is that all gates are diagonal in the $\ket{0}\pm\ket{1}$ basis.
Therefore all gates are commutable.
In other words, there is no temporal structure and hence it is called ``instantaneous''.
Let $H$ be the Hadamard gate $\tfrac{1}{\sqrt{2}}\BinMatrix{1}{1}{1}{-1}$.
 If a gate $U$ is diagonal in the $\ket{0}\pm\ket{1}$ basis,
there exists a diagonal matrix $D$ such that $M_U= H^{\otimes n} D H^{\otimes n}$.
Moreover $H$ is its own inverse; That is, $HH=I_2$.
Any two $H$'s between each pair of gates cancel.
This leads to an alternative view of \IQP\ circuit in which
each qubit line starts and ends with an $H$ gate and all gates in between are diagonal.

\begin{definition}
  An \IQP\ circuit on $n$ qubit lines is a quantum circuit with the following structure:
  each qubit line starts and ends with an $H$ gate, and all other gates are diagonal.
\end{definition}

We will focus particularly on $1,2$-local \IQP, which means that every intermediate gate acts on $1$ or $2$ qubits.
It was shown that a classical weak simulation of $1,2$-local \IQP\ with multiplicative error implies the polynomial hierarchy collapse to the third level \cite{BJS11}.
Let $Z=\BinMatrix{1}{0}{0}{-1}$.
The hardness of simulation holds even if we restrict gates to the phase gate $e^{\I(\pi/8)Z}=\BinMatrix{e^{\I\pi/8}}{0}{0}{e^{-\I\pi/8}}$ and 
the controlled $Z$-gate $CZ=\left[
  \begin{smallmatrix}
	1 & 0 & 0 & 0 \\
	0 & 1 & 0 & 0 \\
	0 & 0 & 1 & 0 \\
	0 & 0 & 0 & -1 \\	
  \end{smallmatrix}
\right]$ other than $H$ gates on two ends of each line.
We will show that this class of \IQP\ circuits corresponds to Ising models with complex 
edge interactions and that therefore the strong simulation of these circuits is \numP-hard,
even allowing an error of any factor $\cst>1$.

To show the relationship between these circuits and Ising partition functions, it is convenient to use another set of gates.
Let $\Pgate{\theta}=e^{\I\theta Z}=\BinMatrix{e^{\I\theta}}{0}{0}{e^{-\I\theta}}$
and $\ZZgate{\theta}=e^{\I\theta Z\otimes Z}=\left[
  \begin{smallmatrix}
	e^{\I\theta} & 0 & 0 & 0 \\
	0 & e^{-\I\theta} & 0 & 0 \\
	0 & 0 & e^{-\I\theta} & 0 \\
	0 & 0 & 0 & e^{\I\theta} \\	
  \end{smallmatrix}
\right]$.
Note   from  \eqref{eqn:quantum:marginal}
that we may multiply a gate by any norm~$1$ constant without affecting the outcome  of the gate.
By multiplying by $e^{-\I\pi/4}$, we may decompose $CZ$ as:
\begin{align}
  \label{eqn:CZ}
  e^{-\I\pi/4}
  \left[
  \begin{smallmatrix}
	1 & 0 & 0 & 0 \\
	0 & 1 & 0 & 0 \\
	0 & 0 & 1 & 0 \\
	0 & 0 & 0 & -1 \\	
  \end{smallmatrix}
  \right]
  &=
  \left[
  \begin{smallmatrix}
	e^{\I\pi/8} & 0 & 0 & 0 \\
	0 & e^{-\I\pi/8} & 0 & 0 \\
	0 & 0 & e^{-\I\pi/8} & 0 \\
	0 & 0 & 0 & e^{\I\pi/8} \\	
  \end{smallmatrix}
  \right]^{2}
  \left[
  \begin{smallmatrix}
	e^{\I\pi/8} & 0 & 0 & 0 \\
	0 & e^{-\I\pi/8} & 0 & 0 \\
	0 & 0 & e^{\I\pi/8} & 0 \\
	0 & 0 & 0 & e^{-\I\pi/8} \\	
  \end{smallmatrix}
  \right]^{14}
  \left[
  \begin{smallmatrix}
	e^{\I\pi/8} & 0 & 0 & 0 \\
	0 & e^{\I\pi/8} & 0 & 0 \\
	0 & 0 & e^{-\I\pi/8} & 0 \\
	0 & 0 & 0 & e^{-\I\pi/8} \\	
  \end{smallmatrix}
  \right]^{14}\nonumber \\
  &=\left(\ZZgate{\pi/8}\right)^{2}\left(\Pgate{\pi/8}\otimes I_2\right)^{14}\left(I_2 \otimes \Pgate{\pi/8}\right)^{14}.
\end{align}
Hence we can replace every $CZ$ gate on qubits $\ii,\jj$ by $2$ copies of $\ZZgate{\pi/8}$ on $\ii,\jj$, 14 copies of $\Pgate{\pi/8}$ on qubit $\ii$, 
and 14 $\Pgate{\pi/8}$ on qubit $\jj$.
It is easy to see that $\ZZgate{\pi/8}$ can be replaced by $CZ$ and $\Pgate{\pi/8}$ as well.
We may therefore assume every gate is either $\Pgate{\pi/8}$ on $1$ qubit 
or $\ZZgate{\pi/8}$ on $2$ qubits without changing the computational power of the circuit.
In general we give the following definition.

\begin{definition}\label{def:IQPonetwo}
  An \IQPonetwo{\theta}\ circuit on $n$ qubit lines is a quantum circuit with the following structure:
  each qubit line starts and ends with an $H$ gate, 
  and every other gate is either $\Pgate{\theta}$ on $1$ qubit or $\ZZgate{\theta}$ on $2$ qubits.
  We assume the input state is always $\ket{0^n}$.
\end{definition}

An example \IQPonetwo{\theta}\ circuit is given in Figure~\ref{fig:IQP}.

The relationship between \IQPonetwo{\theta} circuits and Ising models was first observed by Fujii and Morimae \cite{FM}.
These connections will be shown next.
For completeness we include our own proofs, which  have a more combinatorial flavour than the original ones by Fujii and Morimae \cite{FM}.
We introduce the following non-uniform Ising model
which has been studied previously. See, for example~\cite{Sokal}.
Let $G=(V,E)$ be a (multi)graph.
The edge interaction is specified by a function $\varphi:E\rightarrow\mathbb{C}$ and
the external field is specified by a function $\tau:V\rightarrow\mathbb{C}$.
The partition function is defined as
\begin{align}
  \label{eqn:NonuniformIsing}
  \ZIsing{G}{\varphi,\tau}=\sum_{\sigma:V\rightarrow\{0,1\}}\prod_{e=(v_\ii,v_\jj)\in E}\varphi(e)^{\delta(\sigma(v_\ii),\sigma(v_\jj))}\prod_{v\in V}\tau(v)^{\sigma(v)},
\end{align}
where $\delta(x,y)=1$ if $x=y$ and $\delta(x,y)=0$ if $x\neq y$.
We write $\ZIsing{G}{\ybeta,\tau}$ when $\varphi(e)=\ybeta$ is a constant function and similarly $\ZIsing{G}{\varphi,\lambda}$ when $\tau(v)=\lambda$.
Notice that this notation is consistent with \eqref{eqn:Ising}.

We will show that the following problem is related
to \SSimIQP{\cst}{\theta } when $e^{\I\theta}$ is a root of unity.
\prob{\IQPIsing{\cst}{\theta}.}
{A (multi)graph $G$ with
 an edge interaction function $\varphi(-)$ taking value $e^{\I\theta}$ or $e^{-\I\theta}$, 
 and
 an external field function~$\tau$ so that for each vertex~$v$
 there are non-negative integers $a_v$ and~$b_v$ so that
 $\tau(v)=(-1)^{a_v}\left( e^{\I\theta} \right)^{b_v}$ or
 $\tau(v)=(-1)^{a_v}\left( e^{-\I\theta} \right)^{b_v}$. }
{A rational number $p$ such that  $   |\ZIsing{G}{\varphi,\tau}|/\cst \leq p \leq  \cst|\ZIsing{G}{\varphi,\tau}|$.
}

\begin{figure}[t]
  \centering
  \begin{minipage}{.45\textwidth}
    \centering
    \includegraphics[width=.9\linewidth]{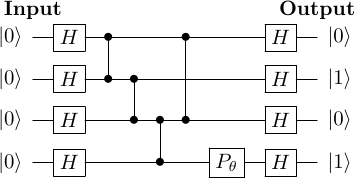}
	\caption{An \IQPonetwo{\theta} circuit. We use two solid dots to denote $R_{\theta}$ gate as it is diagonal and symmetric.}
	\label{fig:IQP}
  \end{minipage}%
  \hspace{0.03\textwidth}
  \begin{minipage}{.45\textwidth}
  \centering
  \includegraphics[width=.9\linewidth]{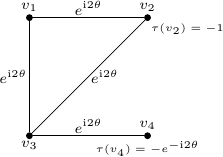}
  \caption{The equivalent Ising instance to the circuit in Figure~\ref{fig:IQP}.}
  \label{fig:IQP:Ising}
  \end{minipage}
\end{figure}

We will first consider inputs to \IQPonetwo{\theta} where $I=[n]$ so all qubits are measured.
Given an \IQPonetwo{\theta}\ circuit $C$ on $n$ qubits and a string $\mathbf{y}\in\{0,1\}^n$,
we can construct a non-uniform Ising instance $G_{C}$ with edge interaction $e^{\I2\theta}$ and external field $\tau_{C;\mathbf{y}}$ such that
\begin{align}
  \label{eqn:IQP:Ising}
  \Pr_C(\mathbf{Y}=\mathbf{y})=2^{-2n}\left|\ZIsing{G_{C}}{e^{\I2\theta},\tau_{C;\mathbf{y}}}\right|^2.
\end{align}
The construction is as follows.
The vertex set $\{v_\ii\}$ contains $n$ vertices and each vertex corresponds to a qubit.
For each gate $\ZZgate{\theta}$ on two qubits $\ii,\jj$, add an edge $(\ii,\jj)$ in $G_{C}$.
For qubit $\ii$, let $p_\ii$ be the number of gates $\Pgate{\theta}$ acting on qubit $\ii$ in $C$.
Let $\tau_{C;\mathbf{y}}(v_\ii)=e^{-\I(2p_{\ii}\theta)}(-1)^{y_\ii}$.
An example of the construction is given in Figure~\ref{fig:IQP:Ising}.

\begin{lemma}
  \label{lem:IQP:Ising}
  Let $C$ be an \IQPonetwo{\theta}\ circuit on $n$ qubits and $\mathbf{y}\in\{0,1\}^n$ be the output.
  Let $G_C$ and $\tau_{C;\mathbf{y}}$ be constructed as above.
  Then \eqref{eqn:IQP:Ising} holds.
\end{lemma}
\begin{proof}
  Suppose $C$ is composed sequentially by  $U_1=H^{\otimes n}$, $U_2$, \dots, $U_{m-1}$, $U_{m}=H^{\otimes n}$,
  where $U_\ii$ is either $\Pgate{\theta}$ on $1$ qubit or $\ZZgate{\theta}$ on $2$ qubits for $2\leq \ii\leq m-1$.
  Notice that $U_1(\mathbf{x},\mathbf{x'})=U_m(\mathbf{x},\mathbf{x'})=2^{-n/2}\prod_{\jj=1}^n(-1)^{x_\jj x_\jj'}$.
  As the input $\ket{\mathbf{x}}=\ket{0^n}$,
  we can rewrite \eqref{eqn:marginal:partitionfunction}:
  \begin{align}
	\label{eqn:marginal:rewrite}
    \Pr_C(\mathbf{Y}=\mathbf{y})
	&=\left|
    \sum_{\sigma:\mathbf{z}\rightarrow\{0,1\} }
    U_1(\mathbf{0},\sigma(\mathbf{z_1}))
    U_m(\sigma(\mathbf{z_{m-1}}),\mathbf{y})
    \prod_{\ii=2}^{m-1} U_\ii(\sigma(\mathbf{z_{\ii-1}}),\sigma(\mathbf{z_\ii}))
    \right|^2\nonumber\\
	&=\left|2^{-n}
    \sum_{\sigma:\mathbf{z}\rightarrow\{0,1\} }
	\prod_{\jj=1}^n(-1)^{0\cdot\sigma(z_{1,\jj})}
	\prod_{\jj=1}^n(-1)^{y_\jj\sigma(z_{m-1,\jj})}
    \prod_{\ii=2}^{m-1} U_\ii(\sigma(\mathbf{z_{\ii-1}}),\sigma(\mathbf{z_\ii}))
    \right|^2\nonumber\\
	&=2^{-2n}\left|
    \sum_{\sigma:\mathbf{z}\rightarrow\{0,1\} }
	\prod_{\jj=1}^n(-1)^{y_\jj\sigma(z_{m-1,\jj})}
    \prod_{\ii=2}^{m-1} U_\ii(\sigma(\mathbf{z_{\ii-1}}),\sigma(\mathbf{z_\ii}))
    \right|^2
  \end{align}
  Let $Q$ denote the quantity inside the norm, that is, 
  \[Q:=\sum_{\sigma:\mathbf{z}\rightarrow\{0,1\} }
	\prod_{\jj=1}^n(-1)^{y_\jj\sigma(z_{m-1,\jj})}
  \prod_{\ii=2}^{m-1} U_\ii(\sigma(\mathbf{z_{\ii-1}}),\sigma(\mathbf{z_\ii})).\]
  Since $U_\ii$'s are diagonal for $2\leq \ii\leq m-1$,
  any configuration $\sigma$ with a non-zero contribution to $Q$ must satisfy that for any $\jj$, $\sigma(z_{1,\jj})=\sigma(z_{2,\jj})=\dots=\sigma(z_{m-1,\jj})$.
  Therefore we may replace $z_{\ii,\jj}$ by a single variable $v_\jj$ for all $1\leq \ii\leq m-1$ so that
  \[Q=\sum_{\sigma:V\rightarrow\{0,1\} }
	\prod_{\jj=1}^n(-1)^{y_\jj\sigma(v_{\jj})}
  \prod_{\ii=2}^{m-1} U_\ii(\sigma(V),\sigma(V)).\]
  Moreover, if $U_\ii$ is the gate $\Pgate{\theta}$ on qubit $\jj$, then $U_\ii(\sigma(V),\sigma(V))=e^{\I\theta}\left(e^{-\I2\theta}\right)^{\sigma(v_\jj)}$.
  If $U_\ii$ is the gate $\ZZgate{\theta}$ on qubits $\jj_1$ and $\jj_2$, then 
  $U_\ii(\sigma(V),\sigma(V))=e^{-\I\theta} \left( e^{\I2\theta} \right)^{\delta(\sigma(v_{\jj_1}),\sigma(v_{\jj_2}))}$,
  where $\delta(x,y)=1$ if $x=y$ and $\delta(x,y)=0$ if $x\neq y$.
  Recall that $p_\jj$ is the number of $\Pgate{\theta}$ gates on qubit $\jj$ and $\tau_{C;\mathbf{y}}(v_\jj)=e^{-\I(2p_\jj\theta)}(-1)^{y_\jj}$.
  Collecting all the contributions, we have
  \begin{align}
	\label{eqn:q:ising}
	Q&=e^{\I(m_1-m_2)\theta}
	\sum_{\sigma:V\rightarrow\{0,1\} }
	\left( e^{\I2\theta} \right)^{m(\sigma)}
	\prod_{\jj=1}^n(-1)^{y_\jj\sigma(v_{\jj})}\left(e^{-\I2\theta}\right)^{p_\jj\sigma(v_\jj)}\nonumber\\
	&=e^{\I(m_1-m_2)\theta}
	\sum_{\sigma:V\rightarrow\{0,1\} }
	\left( e^{\I2\theta} \right)^{m(\sigma)}	
	\prod_{\jj=1}^n\tau_{C;\mathbf{y}}(v_\jj)^{\sigma(v_{\jj})}\\
    &=e^{\I(m_1-m_2)\theta}
	\ZIsing{G_{C}}{e^{\I2\theta},\tau_{C;\mathbf{y}}},\nonumber
  \end{align}
  where $m_\ii$ is the number of $\ii$ qubit(s) gates for $\ii\in\{1,2\}$,
  and, from \eqref{eqn:Ising},
  $m(\sigma)$ is  the number of monochromatic edges under $\sigma$.
  We get \eqref{eqn:IQP:Ising} by substituting \eqref{eqn:q:ising} in \eqref{eqn:marginal:rewrite}.
\end{proof}

Similar results hold when some qubits are not measured.
To show it, we need the following fact. 
It can be viewed as an application of Parsevals's identity on 
the length-$2^n$  vector $\{C_{\mathbf{z}}\}$ indexed by $\mathbf{z}\in\{0,1\}^n$
over an orthonormal basis $\{e_{\mathbf{z}}\}$
where basis element $e_{\mathbf{z}}$ 
has value
$ 2^{-\frac{n}{2}}(-1)^{\mathbf{z}\cdot\mathbf{z}'}$
in position $\mathbf{z}'$.
We include a proof for completeness.

\begin{claim}
  \label{claim:normsum}
  Let $\{C_{\mathbf{z}}\}$ be $2^n$ complex numbers where $\mathbf{z}$ runs over $\{0,1\}^n$.
  Then we have
  \[\sum_{\mathbf{z'}\in\{0,1\}^n}\left|\sum_{\mathbf{z}\in\{0,1\}^n}C_{\mathbf{z}}(-1)^{\mathbf{z}\cdot\mathbf{z'}}\right|^2
  =2^{n}\sum_{\mathbf{z}\in\{0,1\}^n}\left|C_{\mathbf{z}}\right|^2.\]
\end{claim}
\begin{proof}
  Notice that for two complex numbers $A$ and $B$, 
  \begin{align}
	\label{eqn:ABnorm}
	|A+B|^2+|A-B|^2&=\left(|A|^2+|B|^2-2|A||B|\cos\theta\right)+\left(|A|^2+|B|^2+2|A||B|\cos\theta\right)\nonumber\\
	&=2\left( |A|^2+|B|^2 \right)
  \end{align}
  where $\theta$ is the angle from $A$ to $B$.
  Hence we have
  \begin{align*}
	\sum_{\mathbf{z'}\in\{0,1\}^n}
	&\left|\sum_{\mathbf{z}\in\{0,1\}^n}C_{\mathbf{z}}(-1)^{\mathbf{z}\cdot\mathbf{z'}}\right|^2
	=\sum_{\substack{\mathbf{z'}\in\{0,1\}^n\\\rm{s.t.\ } z'_n=0}}
	\left|\sum_{\mathbf{z}\in\{0,1\}^n}C_{\mathbf{z}}(-1)^{\mathbf{z}\cdot\mathbf{z'}}\right|^2+
	\sum_{\substack{\mathbf{z'}\in\{0,1\}^n\\\rm{s.t.\ } z'_n=1}}
	\left|\sum_{\mathbf{z}\in\{0,1\}^n}C_{\mathbf{z}}(-1)^{\mathbf{z}\cdot\mathbf{z'}}\right|^2\\
	&=\sum_{\mathbf{y'}\in\{0,1\}^{n-1}}
	\left|\sum_{\mathbf{y}\in\{0,1\}^{n-1}}C_{\mathbf{y}0}(-1)^{\mathbf{y}\cdot\mathbf{y'}}
	+\sum_{\mathbf{y}\in\{0,1\}^{n-1}}C_{\mathbf{y}1}(-1)^{\mathbf{y}\cdot\mathbf{y'}}\right|^2+\\
	&\sum_{\mathbf{y'}\in\{0,1\}^{n-1}}
	\left|\sum_{\mathbf{y}\in\{0,1\}^{n-1}}C_{\mathbf{y}0}(-1)^{\mathbf{y}\cdot\mathbf{y'}}
	-\sum_{\mathbf{y}\in\{0,1\}^{n-1}}C_{\mathbf{y}1}(-1)^{\mathbf{y}\cdot\mathbf{y'}}\right|^2\\
	&=2\sum_{\mathbf{y'}\in\{0,1\}^{n-1}}
	\left(\left|\sum_{\mathbf{y}\in\{0,1\}^{n-1}}C_{\mathbf{y}0}(-1)^{\mathbf{y}\cdot\mathbf{y'}}\right|^2
	+\left|\sum_{\mathbf{y}\in\{0,1\}^{n-1}}C_{\mathbf{y}1}(-1)^{\mathbf{y}\cdot\mathbf{y'}}\right|^2\right),
  \end{align*}
  where in the last line we apply \eqref{eqn:ABnorm}.
  The claim holds by induction.
\end{proof}

We then have the following reduction.

\begin{lemma}
  \label{lem:ssimIQP:Ising}  
  Let $\cst>1$ and
  $\theta\in[0,2\pi)$. 
  Then  
  
  \centering{
  {\SSimIQP{\cst}{\theta }}  \Reduce{T}   {\IQPIsing{\cst^2}{2\theta }}.
  }
\end{lemma}

\begin{proof}
If all qubits in the input to 
\SSimIQP{\cst}{\theta } are measured, then  the   result follows from Lemma\ \ref{lem:IQP:Ising}.
Otherwise, without loss of generality we assume the first $n-\ss$ qubits are measured.
Let  $C$, $I=[n-\ss]$ and $\mathbf{y'}\in\{0,1\}^{n-\ss}$ be the  
input to \SSimIQP{\cst}{\theta }.
We use \eqref{eqn:quantum:marginal:partial}, \eqref{eqn:marginal:rewrite}, and the first line of \eqref{eqn:q:ising}:
\begin{align}
  \label{eqn:q:part}
  \Pr_{C;I}(\mathbf{Y'}=\mathbf{y'})
    &=\sum_{\mathbf{z'}\in\{0,1\}^\ss}\Pr_{C}(\mathbf{Y}=\mathbf{y'}\mathbf{z'})\nonumber\\
    &=2^{-2n}\sum_{\mathbf{z'}\in\{0,1\}^\ss}\Bigg|
	\sum_{\sigma:V\rightarrow\{0,1\} }
	\left( e^{\I2\theta} \right)^{m(\sigma)}
	\left(\prod_{l=n-\ss+1}^{n}(-1)^{z'_{l-(n-\ss)}\sigma(v_{l})}\left(e^{-\I2\theta}\right)^{{p_l}\sigma(v_{l})}\right)\nonumber\\
	&\left(\prod_{\jj=1}^{n-\ss}(-1)^{y_\jj'\sigma(v_{\jj})}\left(e^{-\I2\theta}\right)^{p_\jj\sigma(v_\jj)}\right)\Bigg|^2\nonumber\\
	&=2^{-2n}\sum_{\mathbf{z'}\in\{0,1\}^\ss}\left|\sum_{\mathbf{z}\in\{0,1\}^\ss}Q_{\mathbf{z}}(-1)^{\mathbf{z}\cdot\mathbf{z'}}\right|^2,
\end{align}
where for $\mathbf{z}\in\{0,1\}^\ss$, $Q_{\mathbf{z}}$ is the contribution of assigning $z_{l-n+\ss}$ to $v_l$ without the possible $-1$ external field, that is,
\[Q_{\mathbf{z}}=	
  \prod_{l=n-\ss+1}^{n}\left(e^{-\I2\theta}\right)^{z_{l-n+\ss}p_l}
  \sum_{\substack{\sigma:V\rightarrow\{0,1\}\mathrm{\ such\ that\ }\\\mathrm{for\ }n-\ss+1\leq l\leq n, \sigma(v_l)=z_{l-n+\ss}}}
  \left( e^{\I2\theta} \right)^{m(\sigma)}
  \prod_{\jj=1}^{n-\ss}(-1)^{y_\jj'\sigma(v_{\jj})}\left(e^{-\I2\theta}\right)^{p_\jj\sigma(v_\jj)}.\]
Apply Claim\ \ref{claim:normsum} on \eqref{eqn:q:part}:
\begin{align}
  \label{eqn:q:partial}
  \Pr_{C;I}(\mathbf{Y'}=\mathbf{y'}) &=2^{-2n+\ss}\sum_{\mathbf{z}\in\{0,1\}^\ss}\left|Q_{\mathbf{z}}\right|^2.
\end{align}
Moreover we have
\begin{align*}
  \left|Q_{\mathbf{z}}\right|^2
  =\left| \sum_{\substack{\sigma:V\rightarrow\{0,1\}\mathrm{\ such\ that\ }\\\mathrm{for\ }n-\ss+1\leq l\leq n, \sigma(v_l)=z_{l-n+\ss}}}
  \left( e^{\I2\theta} \right)^{m(\sigma)}
  \prod_{\jj=1}^{n-\ss}(-1)^{y_\jj'\sigma(v_{\jj})}\left(e^{-\I2\theta}\right)^{p_\jj\sigma(v_\jj)} \right|^2.
\end{align*}

\begin{figure}[t]
  \centering
  \includegraphics[width=.6\linewidth]{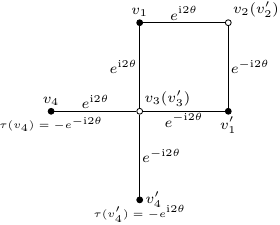}
  \caption{The equivalent Ising instance to the circuit in Figure~\ref{fig:IQP}, if qubits $2$ and $3$ are unmeasured.
  The notation $v_2 (v'_2)$ indicates that vertices $v_2$ and $v'_2$ have been identified.}
  \label{fig:IQP:Ising:2}
\end{figure}

We construct the following instance of \IQPIsing{\cst^2}{2\theta }.
We first construct $G_C=(V,E)$ with edge interaction $e^{\I2\theta}$ as before.
The vertex set $\{v_j\}$ contains one vertex for each of the $n$~qubits. For each gate 
$\ZZgate{\theta}$ on two qubits $\ii,\jj$ we add edge $(\ii,\jj)$
with edge interaction $e^{\I 2 \theta}$ to~$G_C$. Now
make a copy $G_C'=(V',E')$ such that the edge interaction is $\overline{e^{\I2\theta}}=e^{-\I2\theta}$.
Let $\varphi_{C;I}$ be this edge interaction function.
Then we identify vertices $v_l$ with $v_l'$ for all $n-\ss+1\leq l\leq n$.
Let $U$ be the set of these identified vertices and let $V_1=V-U$ and $V_1'=V'-U$.
The external field $\tau=\tau_{C;I,\mathbf{y'}}$ is defined as follows:
for any $v\in U$, $\tau(v)=1$; for any $v_\ii\in V_1$, $\tau(v_\ii)=e^{-\I(2p_{\ii}\theta)}(-1)^{y'_\ii}$; 
and for any $v'_\ii\in V_1'$, $\tau(v'_\ii)=\overline{\tau(v_\ii)}=e^{\I(2p_{\ii}\theta)}(-1)^{y'_\ii}$.
Informally, this instance was formed by putting
$G_C$ and its complement together and identifying vertices that correspond to unmeasured qubits.
Note that if two vertices in~$U$ are connected by an edge,
then they are actually connected by two edges, and the product of the two edge interactions is~$1$.
We therefore remove all edges with both endpoints in~$U$.
Call the resulting graph $H_C$.
One can verify that $(H_C,\varphi_{C;I},\tau_{C;I,\mathbf{y'}})$ is a valid instance of \IQPIsing{\cst^2}{2\theta}.
An example of the construction is given in Figure~\ref{fig:IQP:Ising:2}.

Fix an assignment $\mathbf{z}\in\{0,1\}^{\ss}$ on $U$.
The contribution $Z_\mathbf{z}$ to \ZIsing{H_C}{\varphi_{C;I},\tau_{C;I,\mathbf{y'}}} can be counted in two independent parts, $V$ and $V'$.
Hence we have
\begin{align*}
  Z_\mathbf{z}=&
  \left(
  \sum_{\sigma_1:V_1\rightarrow\{0,1\}} \left( e^{\I2\theta} \right)^{m_\oo(\sigma_1,\mathbf{z})} \prod_{\ii=1}^{n-\ss}\tau(v_\ii)^{\sigma(v_\ii)}
  \right)\cdot
  \left(
  \sum_{\sigma_1':V_1'\rightarrow\{0,1\}} \left( e^{-\I2\theta} \right)^{m_\oo'(\sigma_1',\mathbf{z})} \prod_{\ii=1}^{n-s}\overline{\tau(v_\ii)}^{\sigma(v_\ii')}
  \right)\\
  =&\left| \sum_{\sigma_1:V_1\rightarrow\{0,1\}} \left( e^{\I2\theta} \right)^{m_\oo(\sigma_1,\mathbf{z})} \prod_{\ii=1}^{n-\ss}\tau(v_\ii)^{\sigma(v_\ii)} \right|^2,
\end{align*}
where given the configurations $\sigma_1$ (or $\sigma_1'$), $m_\oo(\sigma_1,\mathbf{z})$ (or $m_\oo'(\sigma_1',\mathbf{z})$) 
is the number of monochromatic edges
with at least one endpoint
in~$V$ (or~$V'$).
Recall that $\tau(v_\ii)=e^{-\I(2p_{\ii}\theta)}(-1)^{y'_\ii}$.
Comparing $Z_\mathbf{z}$ to $\left|Q_{\mathbf{z}}\right|^2$,
the only difference is that in $\left|Q_{\mathbf{z}}\right|^2$, $e^{\I2\theta}$ is raised to the number of monochromatic edges in the whole $V$ instead of $V_1$.
However for any monochromatic edge in $U$, its contribution is independent from the configuration $\sigma$, and hence can be moved outside of the sum.
All such terms are cancelled after taking the norm.
This implies $Z_\mathbf{z}=\left|Q_{\mathbf{z}}\right|^2$.
Therefore \eqref{eqn:q:partial} can be rewritten as
\begin{align}
  \label{eqn:marginal:partial:Ising}
  \Pr_{C;I}(\mathbf{Y'}=\mathbf{y'})&=2^{-2n+\ss}\sum_{\mathbf{z}\in\{0,1\}^\ss}Z_{\mathbf{z}}\nonumber\\
  &=2^{-2n+\ss}\ZIsing{H_C}{\varphi_{C;I},\tau_{C;I,\mathbf{y'}}}
  =2^{-2n+\ss}|\ZIsing{H_C}{\varphi_{C;I},\tau_{C;I,\mathbf{y'}}}|.
\end{align}
The lemma follows from the above equation.
\end{proof}
\begin{remark}
  In fact, the construction of $H_C$ can be further simplified.
  If $v\in V$ and $v'\in V'$ connect to some $u\in U$, we can replace edges $(u,v)$ and $(u,v')$ 
  by a new edge $(v,v')$ with an Ising interaction $\tfrac{2}{e^{\I4\theta}+e^{-\I4\theta}}$. 
  (In case $e^{\I4\theta}+e^{-\I4\theta}=0$ this interaction is equality and we identify $v$ with $v'$.)
  Therefore we can reduce an instance of \SSimIQP{\cst}{\theta} to an Ising model of size linear in $|I|$, the number of measured qubits.
  If $|I|=O(\log n)$, then the reduced Ising instance is tractable and so is the simulation.
  This matches the strong simulation result by  Shepherd 
  (see \cite[Theorem 3.4]{BJS11} , the remark following that theorem  and also \cite{Shepherd10}.)
\end{remark}

The reduction also works in the other direction when $e^{\I\theta}$ is a root of unity.

\begin{theorem}
  \label{thm:IQP:Ising}
  Let $e^{\I\theta}$ be a root of unity and let $\cst>1$.
  Then  
   
	\centering
	\IQPIsing{\cst}{2\theta}\ $\equiv_{\mathrm{T}}$
	\SSimIQP{\cst^{1/2}}{\theta}.
  
\end{theorem}
\begin{proof}
  Lemma~\ref{lem:ssimIQP:Ising} implies a reduction from the right hand side to the left hand side.
  In the rest of the proof we show the other direction.
  As $e^{\I\theta}$ is a root of unity, there exists a positive integer $\tt$ such that $e^{-\I2\theta}=e^{\I2\tt\theta}$.
  Given an instance $(G,\varphi,\tau)$ of \IQPIsing{\cst}{2\theta},
  we may replace each edge of interaction $e^{-\I2\theta}$ by $\tt$ parallel edges of weight $e^{\I2\theta}$.
  Moreover, we may assume the external field is of the form $\tau(v_\ii)=(-1)^{a_\ii}\left(e^{-\I2\theta}\right)^{b_\ii}$ for the same reason.
 
  We construct an \IQPonetwo{\theta} circuit $C$ on $n=|V|$ qubits.
  For each edge $(v_\ii,v_\jj)\in E$, we add a quantum gate $\ZZgate{\theta}$ on qubits $\ii$ and $\jj$.
  For each $1\leq \ii\leq n$, we add $b_\ii$ many quantum gate $\Pgate{\theta}$ on qubits $\ii$ and let the output $y_\ii=1$ on qubit $\ii$ if $a_\ii$ is odd.
  By Lemma~\ref{lem:IQP:Ising} we see that $2^{2n}\Pr_C(\mathbf{Y}=\mathbf{y})=\left|\ZIsing{G}{e^{\I2\theta},\tau}\right|^2$.
\end{proof}

Suppose the Ising instance in the proof of Theorem~\ref{thm:IQP:Ising}
has no external field and  has a constant edge interaction $e^{\I2\theta}$.
Then it is not hard to see that the above construction does not rely on $e^{\I\theta}$ being a root of unity and works for general $\theta$.
Hence we have the following lemma.

\begin{lemma}
  \label{lem:NIsingtoIQP}
  Let $e^{\I\theta}\in\mathbb{C}$ and $\cst>1$.
  Then  
  
	\centering
	\CNIsing{\cst}{e^{\I\theta} }\ \Reduce{T}\ \SSimIQP{\cst^{1/2}}{\theta/2 }.
  
\end{lemma}

We can now prove  our main result about IQP.
\setcounter{counter:save}{\value{theorem}}
\setcounter{theorem}{\value{counter:IQP}}
\begin{theorem} 
  Suppose $\cst>1$ and $\theta\in(0,2\pi)$.
  If  $e^{\I \theta}$ is an algebraic complex number
  and $e^{\I8\theta}\neq 1$ then  
  \SSimIQP{\cst}{\theta }\ is \numP-hard.
\end{theorem}
\setcounter{theorem}{\value{counter:save}}

\begin{proof} 
This follows from
Lemma~\ref{lem:NIsingtoIQP} and Corollary \ref{cor:norm1:hard}.
\end{proof} 

We note that if $e^{\I8\theta} = 1$, then \SSimIQP{\cst}{\theta } has a polynomial time algorithm.
By Theorem \ref{thm:IQP:Ising}, 
\SSimIQP{\cst}{\theta } can be reduced to \IQPIsing{\cst^2}{2\theta}.
If $e^{\I8\theta} = 1$, then $e^{\I2\theta}$ is an integer power of $\I$.
Therefore both the edge weight and the vertex weight of \IQPIsing{\cst^2}{2\theta}\ are powers of $\I$.
The algorithm from \cite{CLX14} (affine-type) can be used to solve \IQPIsing{\cst^2}{2\theta}.
See also case 1 of Theorem \ref{thm:fields}.

In  a related result, 
Bremner et al.\ \cite[Corollary 3.3]{BJS11} showed that weakly simulating \IQP\ with multiplicative error implies that the polynomial hierarchy collapses to the third level.
More precisely, their result is the following.
Suppose $C$ is an \IQPonetwo{\pi/8} circuit on $n$ qubits.
If there exists a classical randomized polynomial time procedure to sample a binary string $\mathbf{Z}$ of length $n$,
such that for every string $\mathbf{y}\in\{0,1\}^n$ and any constant $1\leq \cst<\sqrt{2}$,
\[ \Pr_C(\mathbf{Y}=\mathbf{y})/\cst\leq\Pr(\mathbf{Z}=\mathbf{y})\leq \cst\Pr_C(\mathbf{Y}=\mathbf{y}),\]
then the polynomial hierarchy collapses to the third level.
The usual measure for determining the quality of a sampling procedure
is  total variation distance. 
The notion of total variation distances is weaker than ``multiplicative error''
so the result in \cite{BJS11} does not rule out weak simulation with small variation distance.
To see this, note that,
if the multiplicative error is $\cst$, then obviously the total variation distance is at most $\cst-1$. 
On the other hand, 
consider two distributions supported by two $n$-bit Boolean strings.
A sample from the first distribution is obtained uniformly choosing
each of the $n$ bits. 
A sample from the second distribution is obtained by
uniformly choosing each of the first $n-1$ bits.
The last bit is $1$ if all other bits are $0$, and is chosen uniformly otherwise.
The total variation distance is $2^{-n}$, but the multiplicative error is infinity at the all $0$ string.
Note that the complexity implication
``polynomial hierarchy collapses to the third level''
is apparently weaker than the consequence of strong simulation from Theorem~\ref{thm:IQP:hardness}, which 
is $\FP=\numP$.

Strong simulation is also studied with respect to other classes of quantum circuits, see for example \cite{JN14}.
The allowable error is usually taken to be additive and exponentially small,
instead of the constant factors that we have studied here. 
For example, \cite{JN14} requires that the output be computed with $k$ bits of precision
in an amount of time that is polynomial in both~$k$ and the size of the input.
Additive error is quite different from multiplicative error.
Also, the amount of accuracy is important.  
Lemma~\ref{lem:duplicate}
shows that
there is no difference between a constant factor and
an FPRAS scenario, in which the error is allowed to be a factor of $1\pm 1/\lgeps$  
for a unary input~$\lgeps$. On the other hand,
achieving a multiplicative error of $1\pm 1/\exp(\lgeps)$ is an entirely different matter.

\section{\BQP\ and the Tutte polynomial}
\label{sec:BQP}
      
Bordewich et al.~\cite{BFLW} raised
the question ``of determining whether the Tutte polynomial is
greater than or equal to, or less than zero at a given point.''     
Thus, they raised the question of determining the complexity of \SignReTutte{x,y}.
In fact, they were especially interested in the case
$x=-t$, $y=- t^{-1}$ where $t= \exp(2 \pi \I/5)$.

We next show that resolving this case is a simple corollary of our results.
After that, we will discuss the motivation for
considering this point $(x,y)$ and its connection to the complexity class \BQP.
We will also briefly discuss a  relevant general result of Kuperberg~\cite{Kuperberg},
which  resolves similar questions by using three results about quantum computation ---
the Solovay-Kitaev theorem, the FLW density theorem, and a result of Aaronson.

Motivated by connections to quantum computing,
we consider the difficulty of the problem \SignReTutte{x,y} when $xy=1$.
 In particular, we study the points $$(x,y) = (\exp(-a \pi \I/b),\exp(a\pi \I/b)),$$ where $a$ and $b$ are positive
integers. 
If $a \in \{0,b/2,b,3 b/2\}$ then the problem is trivial
since $(x,y)$ is one of the so-called ``special points''
($(1,1)$, $(-1,-1)$, $(-i,i)$ and $(i,-i)$) where evaluating the Tutte polynomial is in \FP~\cite{JVW}.
We can assume without loss of generality that $a<2b$ since adding   $2\pi$ to the
argument of a complex number doesn't change anything. 
We can now prove the main result of this section.

\setcounter{counter:save}{\value{theorem}}
\setcounter{theorem}{\value{counter:BQP}}
\begin{theorem}
 Consider the point $(x,y) = (\exp(-a \pi \I/b),\exp(a \pi \I/b))$,
where $a$ and $b$ are positive integers 
satisfying $0< a/b < 2$  and $a\not\in\{b/2,b,3b/2\}$.
If $a$ is odd and  $\cos(a \pi/b) < 11/27$
then \NonzeroSignReTutte{x,y} is $\numP$-hard.
Thus \SignReTutte{x,y} is also $\numP$-hard.
\end{theorem} 
\setcounter{theorem}{\value{counter:save}}     
     
\begin{proof}
 
 We will use the fact that 
$$q = (x-1)(y-1) = 2 - x - y = 2 - \exp(-a \pi \I/b) - \exp(a \pi \I/b) = 
2 - 2 \cos(a \pi /b),$$
which is real. 
Since $0<a/b<2$ and $a\not\in\{b/2,b,3b/2\}$,
$q\in(0,4)$ and $q\neq 2$. 
 
We implement $(x',y')$ using a $b$-thickening from $(x,y)$.
Then, since $a$ is an odd positive integer,
$$y' = y^b =  
\exp(a \pi \I)=-1.$$  
So $x' = 1+q/(y'-1) = 1 - q/2 = \cos(a \pi/b)$.

Now since $x' < 11/27$, 
\cite[Theorem 1, Region F]{GJSign} shows
that computing the sign of $\ZTutte{-}{q,y'-1}$
is \numP-hard.
As we showed in the argument that established Lemma~\ref{lem:nozeroadd} (see the paragraph before the statement of the lemma),
the same is true if the oracle returns any answer when the value is~$0$.

Since $x'$ and $y'$ are not~$1$, \eqref{eq:TutteRC} 
shows that it is also hard to compute the sign of $T(x',y')$.
The result now follows from Observation~\ref{obs:stretchthicken}.   
\end{proof}
 
 Since $-\exp(-2 \pi \I/5) = \exp(\pi \I) \exp(-2 \pi \I/5) = \exp(3 \pi \I/5)$, we 
 can take $a=3$ and $b=5$ to
 obtain Corollary~\ref{cor:five}, which says that  \NonzeroSignReTutte{1/y,y} is $\numP$-hard 
 for $y=-\exp(-2\pi \I/5)$.

 Theorem~\ref{thm:BQP}  is very close to a special case of the following result of Kuperberg.
A  \emph{link} is a collection of smooth simple closed curves embedded in
$3$-dimensional space. 
 $V_L(t)$ denotes the Jones polynomial of
a link~$L$ evaluated at point~$t$. We do not need the detailed definition of the Jones polynomial in
order to state Kuperberg's theorem.

\begin{theorem} \cite[Theorem 1.2]{Kuperberg}\label{thm:Kup}
Let $V(L,t)$ be the Jones polynomial of a link~$L$ described by a link diagram, and let $t$
be a principal root of unit other
than 
$\exp(2 \pi \I/r)$ where $r\in \{1,2,3,4,5,6\}$. Let $0<A<B$ be two positive real numbers and assume as
a promise that either $|V(L,t)|<A$ or $|V(L,t)|>B$. Then it is \numP-hard to decide which inequality holds.
Moreover, it is still \numP-hard when $L$ is a knot.
\end{theorem}

The connection is as follows.  
There is a result of Thistlethwaite~\cite{Thistlethwaite} (see \cite[(6.1)]{JVW}),
showing that when $L$ is an alternating link with associated planar graph~$G(L)$, 
then $V_L(t) = f_L(t) T(G;-t,-t^{-1})$,
where $f_L(t)$ is an easily-computable factor which is plus or minus a half integer power of $t$.
Thus, the evaluation of Jones polynomial of an alternating link is an easily-computable multiple of an evaluation of
the Tutte polynomial along the hyperbola $x y=1$ (where, for some value~$t$, $x=-t$ and $y=-t^{-1}$), 
as in Theorem~\ref{thm:BQP}.
The importance of these evaluations is established in \cite[Theorem~6.1]{BFLW} which shows that
all of the problems in the quantum complexity class \BQP\ 
(consisting of those decisions problems that can be solved by a quantum computer in polynomial time)
can also be solved classically in polynomial time using an oracle 
that returns the sign of the real part of the Jones polynomial of a link,
evaluated at the point~$t=\exp(2 \pi \I/5)$ (the point studied in Corollary~\ref{cor:five}). 

Kuperberg's theorem (Theorem~\ref{thm:Kup})
is  incomparable to Theorem~\ref{thm:BQP}.
In some respects, Theorem~\ref{thm:Kup} is more general ---
  it does not have the restriction $\cos(a \pi/b) < 11/27$.
 Also,
$G(L)$ is always planar, which is essential for the connection to \BQP,
and it applies to a wide range of~$A$ and~$B$.
On   the other hand,   the most relevant case $A=B=0$
(the one that relates to the \BQP~result of \cite{BFLW})
is actually excluded from
Theorem~\ref{thm:Kup} since~$A$ and~$B$ must be different and positive.
We are not sure whether Kuperberg's proof can be adapted to
include this case, where the goal would be to determine whether
$|V(L,t)| \geq 0$ or $|V(L,t)| \leq 0$. This is covered by Theorem~\ref{thm:BQP}.

In any case, it seems interesting to note that the proof of  Theorem~\ref{thm:BQP}
is combinatorial (about Tutte polynomials only) whereas
the proof of Theorem \ref{thm:Kup} is essentially about quantum computation.
(Kuperberg describes it as ``a mash-up of three standard theorems in quantum computation''.)
        
We refer the reader to \cite{DA} for more recent results   
giving \BQP-hardness of  \emph{multiplicative} approximations of the Jones polynomial
of the plat closure of a braid at roots of unity.
Also, we note that other works such as \cite{GL10} have suggested the idea of using
tractable planar evaluations of these polynomials to give efficient classical simulations for
special cases of quantum circuits.

\section{Ising with a field}

In Section~\ref{sec:field}, we will extend our Ising hardness results from Theorems \ref{thm:main} and \ref{thm:relaxed} 
to the situation in which we have an external field~$\lambda\neq 1$. 
To obtain our hardness results, we need a lower bound on the relevant partition functions.

\subsection{Lower bounds on partition functions}
\label{sec:LB}

Suppose we have two edge weights $\ybeta$ and $\ybeta'$ that are close.
It is easy to bound the distance between $\ZIsing{G}{\ybeta}$ and $\ZIsing{G}{\ybeta'}$ additively, but not multiplicatively.
To convert an absolute error into a relative error, one needs some lower bound on the partition function.
However, when the edge interaction $\ybeta$ is negative or complex, 
it is possible that the partition function vanishes.
Assuming that it doesn't  vanish, we would like to know how close to zero could it get.
When $\ybeta$ is rational, an exponential lower bound is easy to obtain by a simple granularity argument, 
but  the argument is more difficult when $\ybeta$ is not rational.
In this section we  give an exponential lower bound which is
valid  when $\ybeta$ is an algebraic number.
The techniques that we use are standard in transcendental number theory, see e.g.~\cite{Bug04}.

We begin with some basic definitions from \cite{Bug04}.
For a polynomial with complex coefficients 
\begin{align*}
  P(x)=\sum_{i=0}^n a_i x^i=a_n\prod_{i=1}^n(x-\alpha_i),
\end{align*}
the (naive)  \emph{height} of $P(x)$ is defined as $\Height{P}:=\max_i\{|a_i|\}$.
A more advanced tool, its  \emph{Mahler measure}, is defined as
\begin{align*}
  \Mahler{P}:=|a_n|\prod_{i=1}^n\max\{1,|\alpha_i|\}.
\end{align*}

There is a standard inequality relating these two measures.
It is proved for complex polynomials in~\cite[Lemma\ A.2]{Bug04}.
For completeness, we include the   proof (following~\cite{Bug04}) for the case 
in which~$P(x)$ is a real polynomial,
which is all that we require.
\begin{lemma}
  Let $P(x)$ be a non-zero real polynomial of degree $n$.
  Then $\Mahler{P}\leq\sqrt{n+1}\ \Height{P}$.
  \label{lem:Mahler}
\end{lemma}
\begin{proof}
  First apply Jensen's formula on $P(x)$ and on the unit circle in the complex plane,
  \begin{align*}
	\Mahler{P}=\exp\left\{ \int_0^1\log |P(e^{2\I\pi t})|{\rm d} t \right\}.
  \end{align*}
  The convexity of exponential functions implies
  \begin{align*}
	\Mahler{P}&\leq \int_0^1 |P(e^{2\I\pi t})|{\rm d} t\leq \left(\int_0^1 |P(e^{2\I\pi t})|^2{\rm d} t\right)^{1/2},
  \end{align*}
  where the second inequality follows by the Cauchy-Schwarz inequality writing $P(x)$ as $f(x)g(x)$ where $g(x)=1$.
  The inner integral yields
  \begin{align*}
	\int_0^1 |P(e^{2\I\pi t})|^2{\rm d} t & = \int_0^1 \left( \left(\sum_{j=0}^n a_j \cos(j\cdot2\pi t)\right)^2 
	+\left(\sum_{j=0}^n a_j \sin(j\cdot2\pi t)\right)^2 \right) {\rm d} t\\
	&=\sum_{i=0}^n a_i^2 + 2 \int_0^1\sum_{0\leq j < k \leq n}a_j a_k (\cos(j\cdot2\pi t)\cos(k\cdot2\pi t)+\sin(j\cdot2\pi t)\sin(k\cdot2\pi t)){\rm d}t\\
	&=\sum_{i=0}^n a_i^2 + 2 \sum_{0\leq j < k \leq n}a_j a_k\int_0^1 \cos( (j-k)\cdot2\pi t){\rm d}t=\sum_{i=0}^n a_i^2.
  \end{align*}
  The claim holds as $\Mahler{P}\leq\left(\sum_{i=0}^n a_i^2\right)^{1/2}\leq \sqrt{n+1}\ \Height{P}$.
\end{proof}

Let $\ybeta\in\mathbb{C}$ be an algebraic number 
and its minimal polynomial over $\mathbb{Z}$ is $P_\ybeta(x)$.
The degree of $P_\ybeta(x)$ is called the  \emph{degree} of $\ybeta$ 
and $\Height{P_\ybeta}$ is called the  \emph{height} of $\ybeta$, also denoted $\Height{\ybeta}$.

We also need the following notion of  \emph{resultants}.
\begin{definition}
  Let $P(x)=a_n\prod_{i=1}^n(x-\alpha_i)$ and $Q(x)=b_m\prod_{i=1}^m(x-\ybeta_i)$ 
  be two non-constant polynomials.
  The resultant of $P(x)$ and $Q(x)$ is defined as
  \begin{align*}
	\Res{P}{Q}&=a_n^m b_m^n\prod_{1\leq i\leq n}\prod_{1\leq i\leq m}(\alpha_i-\ybeta_j).
  \end{align*}
\end{definition}
It is a standard result that $\Res{P}{Q}$ is an integer polynomial in the coefficients of $P(x)$ and $Q(x)$.
The resultant is also the determinant of the so-called Sylvester matrix.
In particular, when $P(x)$ and $Q(x)$ are integer polynomials, $\Res{P}{Q}$ is always an integer,
as the Sylvester matrix is an integer matrix in this case.
Moreover, we can rewrite the resultant as follows:
\begin{align*}
  \Res{P}{Q}&=a_n^m\prod_{1\leq i\leq n}Q(\alpha_i)=(-1)^{mn}b_m^n\prod_{1\leq j\leq m}P(\ybeta_j).
\end{align*}

Now we are ready to give a lower bound for any integer polynomial evaluated at an algebraic number.
It is a standard result in algebraic number theory.
For completeness we provide a proof here and the treatment is from \cite[Theorem~A.1]{Bug04}.
\begin{lemma}
  Let $P(x)$ be an integer polynomial of degree $n$,
  and $\ybeta\in\mathbb{C}$ be an algebraic number of degree $d$.
  Then either $P(\ybeta)=0$ or 
  \begin{align*}
	|P(\ybeta)|&\geq C_\ybeta^{-n} \left( (n+1)\Height{P} \right)^{-d+1}.
  \end{align*}
  where $C_\ybeta>1$ is an effectively computable constant that only depends on $\ybeta$.
  \label{lem:alge:poly:lb}
\end{lemma}
\begin{proof}
  Assume $P(\ybeta)\neq 0$.
  Let $Q(x)=b_d\prod_{i=1}^d(x-\ybeta_i)$ be the minimal polynomial of $\ybeta$ over $\mathbb{Z}$ with $\ybeta_1=\ybeta$.
  
  Suppose there is an $j\neq 1$ such that $P(\ybeta_j)=0$.
  As $Q(x)$ is the minimal polynomial of $\ybeta$, none of $\ybeta_j$ could be a rational number.
  Hence there is an automorphism of the splitting field of $Q(x)$ that maps $\ybeta_j$ to $\ybeta$.
  Applying this automorphism on both sides of $P(\ybeta_j)=0$,
  we get $P(\ybeta)=0$.
  Contradiction!
  
  Hence we have $P(\ybeta_i)\neq 0$ for all $i$ and the resultant of $P(x)$ and $Q(x)$ is non-zero.
  Since $\Res{P}{Q}$ is an integer, we have
  \begin{align*}
	1\leq |\Res{P}{Q}| = |b_d|^n\prod_{1\leq i \leq d} |P(\ybeta_i)|.
  \end{align*}
  Clearly, by triangle inequality we have $|P(\ybeta_i)|\leq (n+1)\Height{P}(\max\{1,|\ybeta_i|\})^n$.
  It implies,
  \begin{align*}
	1&\leq |P(\ybeta)| |b_d|^n\left( (n+1)\Height{P} \right)^{d-1}\prod_{2\leq i \leq d} (\max\{1,|\ybeta_i|\})^n\\
	&=|P(\ybeta)|\left( (n+1)\Height{P} \right)^{d-1} \left( \frac{\Mahler{Q}}{\max\{1,|\ybeta|\}} \right)^n\\
	&\leq |P(\ybeta)|\left( (n+1)\Height{P} \right)^{d-1} \left(\sqrt{d+1}\Height{\ybeta}\right)^n
  \end{align*}
  where the last inequality follows from Lemma~\ref{lem:Mahler}.
  Therefore we have 
  \begin{align*}
	|P(\ybeta)|\geq \left( (n+1)\Height{P} \right)^{-d+1} \left(\sqrt{d+1}\Height{\ybeta}\right)^{-n}.
  \end{align*}
  Let $C_\ybeta=\sqrt{d+1}\Height{\ybeta}$ and the lemma holds.
\end{proof}

\begin{lemma}  \label{lem:lowerbound}
  Let $G$ be a graph and $\ybeta\in\mathbb{C}$ a non-zero algebraic number of degree $d$.
  There exists a positive constant $C>1$ depending only on $\ybeta$ such that 
  if $\ZIsing{G}{\ybeta}\neq 0$, then $|\ZIsing{G}{\ybeta}|>C^{-m}$, where $m$ is the number of edges in $G$.
\end{lemma}
\begin{proof}
  Given a graph $G$, first suppose that $G$ is not connected, $G_i$'s are the components of $G$.
  Then \ZIsing{G}{\ybeta}$=\prod_i$\ZIsing{G_i}{\ybeta}.
  It is easy to see that if the claim holds for all components it hold for $G$ as well.
  Therefore in the following we may assume $G$ is connected.
  Then $m\geq n-1$ where $n$ is the number of vertices.

  We can rewrite $\ZIsing{G}{\ybeta}$ as a polynomial in $\ybeta$ as follows,
  \begin{align*}
	P(\ybeta)=\ZIsing{G}{\ybeta}=\sum_{i=0}^m C_\ii \ybeta^\ii,
  \end{align*}
  where $C_\ii$ is the number of configurations such that there are exactly $\ii$ many monochromatic edges.
  Notice that $\sum_{\ii=0}^m C_\ii=2^n$, we have $\Height{P}\leq 2^n$.
  Assume $P(\ybeta)\neq 0$. 
  Apply Lemma~\ref{lem:alge:poly:lb} and we obtain
  \begin{align*}
	|P(\ybeta)|&\geq C_\ybeta^{-m} \left( (m+1)\Height{P} \right)^{-d+1}\\
	&\geq (m+1)^{-d+1}C_\ybeta^{-m} 2^{-(d-1)n},
  \end{align*}
  where $C_\ybeta>1$ is a constant depending only on $\ybeta$.
  As $m\geq n-1$, the right hand side decays exponentially in $m$ and the lemma follows.
\end{proof}

\begin{lemma}  \label{lem:lowerbound:field}
  Let $G$ be a graph and $\ybeta,z\in\mathbb{C}$ two roots of unity.
  Let $n$ be the number of vertices in $G$ and $m$ the number of edges.
  There exists a positive constant $C>1$ depending only on $\ybeta$ and $z$ such that 
  if $\ZIsing{G}{\ybeta,z}\neq 0$, then $|\ZIsing{G}{\ybeta,z}|>C^{-m}$.
\end{lemma}
\begin{proof}
  As in the previous lemma we may assume $G$ is connected and $m\geq n-1$.
  Suppose $y$ is of order $d_1$ and $z$ order $d_2$.
  Let $d$ be the least common multiple of $d_1$ and $d_2$.
  Then there exists a root of unity $w$ of order $d$ such that $y=w^{t_1}$ and $z=w^{t_2}$.
  
  Given a graph $G$, we can rewrite $\ZIsing{G}{\ybeta,z}$ as a polynomial in $\ybeta$ and $z$ as follows,
  \begin{align*}
    \ZIsing{G}{\ybeta,z}=\sum_{\jj=0}^n\sum_{\ii=0}^m C_{\ii,\jj} \ybeta^\ii z^\jj,
  \end{align*}
  where $C_{\ii,\jj}$ is the number of configurations such that there are exactly $\ii$ many monochromatic edges and $\jj$ many $1$ vertices.
  Let
  \begin{align*}
    P(w)=\ZIsing{G}{\ybeta,z}=\sum_{\jj=0}^n\sum_{\ii=0}^m C_{\ii,\jj} w^{t_1 \ii + t_2 \jj}=\sum_{\ell=0}^{t_1 m+t_2 n}C'_\ell w^{\ell},
  \end{align*}
  where $C'_\ell=\sum_{t_1\ii+t_2\jj=\ell}C_{\ii,\jj}$.
  Notice that $\sum_{\ell=0}^{t_1 m+t_2 n}C'_\ell=\sum_{\jj=0}^n\sum_{\ii=0}^m C_{\ii,\jj}=2^n$, we have $\Height{P}\leq 2^n$.
  Assume $P(w)\neq 0$.
  Apply Lemma~\ref{lem:alge:poly:lb} and we obtain
  \begin{align*}
    |P(w)|&\geq C_w^{-t_1 m-t_2 n} \left( (t_1 m+t_2 n+1)\Height{P} \right)^{-d+1}\\
	&\geq (t_1 m+t_2 n+1)^{-d+1}C_w^{-t_1 m-t_2 n} 2^{-(d-1)n},
  \end{align*}
  where $C_w>1$ is a constant depending only on $w$.
  As $m\geq n-1$, the right hand side decays exponentially in $m$ and the lemma follows.
\end{proof}

\subsection{Hardness results}
\label{sec:field}

In this section we will show hardness results when both the edge interaction and external field are roots of unity.

We first consider the external field $-1$.
We describe the edge interaction by 
specifying an interaction matrix 
$\left[\begin{matrix}
    n_{00} & n_{01} \\ n_{10} & n_{11}
\end{matrix}\right]$, 
 where $n_{ij}$ is the weight when the two endpoints have spins $i$ and $j$, respectively.
In this notation, a binary equality is 
$\left[\begin{matrix}
    1 & 0 \\ 0 & 1
\end{matrix}\right]$,
and an Ising interaction with weight $\ybeta$ is 
$\left[\begin{matrix}
    \ybeta & 1 \\ 1 & \ybeta
\end{matrix}\right]$.
Given a gadget with two distinguished vertices,
we may view it as an edge and compute its effective interaction matrix $M$.
Then we say the gadget \emph{implements} $M$.
Also, recall the definitions of $k$-stretch and $k$-thickening (Observation \ref{obs:stretchthicken}, for example).

\begin{lemma}  \label{lem:field:-1}
  Let $\cst>1$ and $\ybeta\in\CC$ be an algebraic complex number such that $\ybeta\neq\pm 1$.
  Then we have \CNonzeroNIsing{\cst}{\ybeta} \Reduce{T} \CNonzeroNIsing{\cst}{\ybeta,-1}.  
\end{lemma}
\begin{proof}
  We first argue that a binary equality can be implemented.
  Consider a $2$-stretch with the edge interaction $y$ and external field $-1$.
  It is easy to calculate that the (effective) interaction matrix is 
  $\left[\begin{smallmatrix}
    y^2-1 & 0 \\ 0 & 1-y^2
  \end{smallmatrix}\right]$. 
  Then do a $2$-thickening.
  The resulting matrix is 
  $\left[\begin{smallmatrix}
    (y^2-1)^2 & 0 \\ 0 & (1-y^2)^2
  \end{smallmatrix}\right]$.
  Up to a constant of $(y^2-1)^2$ this is equality.

  Suppose $G=(V,E)$ is an input to \CNonzeroNIsing{\cst}{\ybeta}.
  We introduce a new vertex $v'$ for every vertex $v\in V$.
  Connect $v$ and $v'$ via this equality gadget, that is, first a $2$-stretch and then a $2$-thickening.
  Hence the external field on $v$ is cancelled with this construction.
  The reduction follows.
\end{proof}

Next we consider the case when a real edge interaction can be implemented. 
If the norm of the interaction is less than $1$, then we can cancel out the external field.

\begin{lemma}
  \label{lem:field:<1norm}
  Let $\cst>1$ and $\cst'>1$.
  Let $\ybeta$ and $z$ be two roots of unity and $z\neq \pm 1$.
  Suppose some real number $w\in (-1,1)$ as an edge interaction is implementable for the Ising model with edge interaction $\ybeta$ and external field $z$.
  Then we have \CNonzeroNIsing{\cst}{\ybeta} \Reduce{T} \CNonzeroNIsing{(\cst \cst')}{\ybeta,z}.
\end{lemma}
\begin{proof}
  Let $G=(V,E)$ be an input to \CNonzeroNIsing{\cst}{\ybeta}.
  Assume $\ZIsing{G}{\ybeta}\neq 0$ as otherwise we are done.
  Suppose $|V|=n$, $|E|=m$, and $V=\{v_i|1\leq i\leq n\}$.

  Suppose $w=0$, which means we can implement inequality (see the remark above Lemma \ref{lem:field:-1}).
  For each vertex $v_i$, we introduce a new vertex $v_i'$ and connect $v_i$ and $v_i'$ by the inequality.
  It is easy to verify that if $v_i$ is assigned $0$, the weight from $v_i$ and $v_i'$ together is $z$;
  when $v_i$ is assigned $1$, the weight is also $z$.
  Hence the external field is effectively cancelled and the reduction follows.

  Otherwise assume $w\neq 0$, that is $w\in(-1,0)\cup(0,1)$.
  For each vertex $v_i$, we introduce a new vertex $v_i'$, and add $2t$ many new edges between $v_i$ and $v_i'$,
  where $t$ is a positive integer which we will choose later.
  By assumption we can implement the edge interaction $w$ and we put it on all new edges.
  Let $V'=\{v_i'|1\leq i\leq n\}$ and we get a new graph $G'=(V\cup V',E')$.
  
  For each vertex $v_i$, the contribution of $v_i$ and $v_i'$ (to the partition function) together is $w^{2t}+z$ when $v_i$ is assigned $0$ 
  and $z(1+w^{2t}z)$ when $v_i$ is assigned $1$.
  Let $\lambda=\tfrac{z(1+w^{2t}z)}{w^{2t}+z}$.
  Notice that $w^{2t}+z\neq 0$ as $|w|<1=|z|$.
  We have 
  \begin{align*}
	\ZIsing{G'}{y,z} = (w^{2t}+z)^n\sum_{\sigma:V\rightarrow\{0,1\}}y^{m(\sigma)}\lambda^{n_1(\sigma)},
  \end{align*}
  where $m(\sigma)$ is the number of monochromatic edges in $E$ under $\sigma$
  and $n_1(\sigma)$ is the number of vertices in $V$ that are assigned $1$.

  Let $Z:=\left|\tfrac{\ZIsing{G'}{y,z}}{(w^{2t}+z)^n}-\ZIsing{G}{\ybeta}\right|$.
  We want to show that $Z$ is exponentially small.
  Apply the triangle inequality:
  \begin{align}\label{eqn:nega:Z}
    |Z| & =\left|\sum_{\sigma:V\rightarrow\{0,1\}}y^{m(\sigma)}(\lambda^{n_1(\sigma)}-1)\right|
	\leq\sum_{\sigma:V\rightarrow\{0,1\}}\left|y^{m(\sigma)}(\lambda^{n_1(\sigma)}-1)\right|\nonumber\\
	& = \sum_{\sigma:V\rightarrow\{0,1\}}\left|\lambda^{n_1(\sigma)}-1\right|
	= \sum_{j=0}^n{n \choose j}\left|\lambda^j-1\right|,
  \end{align}
  where we used the fact that $|y|=1$.
  Let $\alpha=\lambda-1=\tfrac{z(1+w^{2t}z)}{w^{2t}+z}-1=\tfrac{w^{2t}(z^2-1)}{w^{2t}+z}$.
  As $z^2-1\neq 0$ and $w^{2t}+z\neq 0$, $|\alpha|$ is decreasing exponentially in $t$.
  We may pick a positive integer $t=O(\log n)$ such that $n e |\alpha| < 1$.
  Applying the triangle inequality again for each $0\leq j\leq n$, we get
  \begin{align}\label{eqn:nega:lambda}
	|\lambda^j-1| &\ =\ |\sum_{l=1}^j{j\choose l}\alpha^l|
	 \ \leq\ \sum_{l=1}^j{j\choose l}|\alpha^l|\nonumber\\
	&\ =\ (|\alpha|+1)^j-1
	 \ \leq\ (|\alpha|+1)^n-1\nonumber\\
	&\ =\ \sum_{l=1}^n{n\choose l}|\alpha|^l
	 \ \leq\ \sum_{l=1}^n\left( \frac{n e |\alpha|}{l} \right)^l\nonumber\\
    &\leq n^2 e |\alpha|,
  \end{align}
  as $\left(\tfrac{ne|\alpha|}{l}\right)^l$ is decreasing in $l$.
  Plugging \eqref{eqn:nega:lambda} into \eqref{eqn:nega:Z} we have
  \begin{align}\label{eqn:nega:Z:estimate}
	|Z|\leq \sum_{j=0}^n{n \choose j} n^2 e |\alpha| = e 2^n n^2 |\alpha|.
  \end{align}
  Since $\ZIsing{G}{\ybeta}\neq 0$, by Lemma~\ref{lem:lowerbound}, there exists a constant $C_\ybeta>1$ such that $|\ZIsing{G}{\ybeta}|>C_\ybeta^{-m}$.
  Since $|\alpha|$ is decreasing exponentially in $t$, by \eqref{eqn:nega:Z:estimate},
  we may pick an integer $t$ that is polynomial in $n$ (and sufficiently large with respect to~$ \cst'$) such that
  \begin{align}
    |Z|< \frac{\cst'-1}{\cst'} C_\ybeta^{-m} < \frac{\cst'-1}{\cst'} |\ZIsing{G}{\ybeta}|.
	\label{eqn:nega:add:error}
  \end{align}
  By the definition of $|Z|$ and again the triangle inequality we get
  \begin{align*}
	\frac{1}{\cst'}  =
	1 - \frac{\cst'-1}{\cst'} \leq
	\frac{|\ZIsing{G'}{\ybeta,z}|}{|w^{2t}+z|^n|\ZIsing{G}{\ybeta}|} \leq 
	1 + \frac{\cst'-1}{\cst'} \leq
	 \cst'.  
   \end{align*}
  This finishes the proof.
\end{proof}

A similar proof works when the implementable real field has a larger than $1$ norm.
Basically when this is the case we may power the external field $z$.
If $z$ is a root of unity then we could power it to $1$.

\begin{lemma}
  \label{lem:field:>1norm}
  Let $\cst>1$ and $\cst'>1$.
  Let $\ybeta$ and $z$ be two roots of unity and $z\neq \pm 1$.
  Suppose some real number $w\in (-\infty,-1)\cup(1,\infty)$ as an edge interaction is implementable 
  for the Ising model with edge interaction $\ybeta$ and external field $z$.
  Then we have \CNonzeroNIsing{\cst}{\ybeta,z^r} \Reduce{T} \CNonzeroNIsing{(\cst \cst')}{\ybeta,z} for any positive integer $r$.
\end{lemma}
\begin{proof}
  Let $G=(V,E)$ be an input to \CNonzeroNIsing{\cst}{\ybeta,z^r}.
  Assume that $\ZIsing{G}{\ybeta,z^r}\neq 0$ as otherwise we are done.
  Suppose $|V|=n$, $|E|=m$, and $V=\{v_i|1\leq i\leq n\}$.

  For each vertex $v_i$, we introduce $r-1$ many new vertices $v_{i,j}$, and add $2t$ many new edges between $v_i$ and each $v_{i,j}$,
  where $j\in[r-1]$ and $t$ is a positive integer which we will choose later.
  By assumption we can implement the edge interaction $w$ and we put it on all new edges.
  Let $V'=\{v_{i,j}|1\leq i\leq n, 1\leq j\leq r-1\}$ and we get a new graph $G'=(V\cup V',E')$.
  
  For each vertex $v_i$, the contribution of $v_i$ and all $v_{i,j}$ combined is $\left(w^{2t}+z\right)^{r-1}$ when $v_i$ is assigned $0$ 
  and $z\left(1+w^{2t}z\right)^{r-1}$ when $v_i$ is assigned $1$.
  Let $\lambda=\tfrac{z\left(1+w^{2t}z\right)^{r-1}}{\left(w^{2t}+z\right)^{r-1}}$.
  Notice that $w^{2t}+z\neq 0$ as $|w|>1=|z|$.  
  We have 
  \begin{align*}
    \ZIsing{G'}{y,z} =\left(w^{2t}+z\right)^{n(r-1)}\sum_{\sigma:V\rightarrow\{0,1\}}y^{m(\sigma)}\lambda^{n_1(\sigma)},
  \end{align*}
  where $m(\sigma)$ is the number of monochromatic edges in $E$ under $\sigma$
  and $n_1(\sigma)$ is the number of vertices in $V$ that are assigned $1$.

  Let $Z:=\left|\tfrac{\ZIsing{G'}{y,z}}{\left(w^{2t}+z\right)^{n(r-1)}}-\ZIsing{G}{\ybeta,z^r}\right|$.
  As the previous proof we show that $Z$ is exponentially small.
  Apply  the triangle inequality:
  \begin{align}\label{eqn:nega:Z:norm>1}
    |Z| & =\left|\sum_{\sigma:V\rightarrow\{0,1\}}y^{m(\sigma)}(\lambda^{n_1(\sigma)}-z^{rn_1(\sigma)})\right|
	\leq\sum_{\sigma:V\rightarrow\{0,1\}}\left|y^{m(\sigma)}(\lambda^{n_1(\sigma)}-z^{rn_1(\sigma)})\right|\nonumber\\
	& = \sum_{\sigma:V\rightarrow\{0,1\}}\left|\lambda^{n_1(\sigma)}-z^{rn_1(\sigma)}\right|
    = \sum_{j=0}^n{n \choose j}\left|\lambda^j-z^{rj}\right|,
  \end{align}
  where we used the fact that $|y|=1$.
  Let $\alpha=\lambda-z^r=\tfrac{z\left(1+w^{2t}z\right)^{r-1}}{\left(w^{2t}+z\right)^{r-1}}-z^r
  =z\left( (z+\mu)^{r-1}-z^{r-1}\right)$,
  where $\mu=\tfrac{1+w^{2t}z}{w^{2t}+z}-z=\tfrac{1-z^2}{w^{2t}+z}\neq 0$.
  As $z^2-1\neq 0$ and $|w|>1$, $|\mu|$ decreases exponentially in $t$.
  Pick a large enough integer $t$ so that $|\mu|<1$.
  Hence $|\alpha|=|z||(z+\mu)^{r-1}-z^{r-1}|=|\sum_{j=1}^{r-1}{r-1\choose j}\mu^jz^{r-1-j}|\leq \sum_{j=1}^{r-1}{r-1\choose j}|\mu^j|< |\mu|2^{r-1}$ 
  by the triangle inequality.
  As $|\mu|$ decreases exponentially in $t$, so does $|\alpha|$.

  Notice that $|\lambda|=|z^r+\alpha|\leq |z|^r+|\alpha|=1+|\alpha|$.
  Pick $t$ large so that $|\alpha|<1$.
  Applying the triangle inequality again for each $0\leq j\leq n$, we get
  \begin{align}\label{eqn:nega:lambda:norm>1}
    |\lambda^j-z^{rj}| &\ =\ \left|\lambda-z^r\right| \left|\sum_{l=0}^{j-1}\lambda^lz^{r(j-1-l)}\right|
    \ \leq\ |\alpha|\left(\sum_{l=0}^{j-1}\left|\lambda^lz^{r(j-1-l)}\right|\right)\nonumber\\
	&\ =\ |\alpha|\left(\sum_{l=0}^{j-1}\left|\lambda\right|^l\right)
	 \ \leq\ |\alpha|\left(\sum_{l=0}^{j-1}\left(1+|\alpha|\right)^l\right)\nonumber\\
     &\ <\ |\alpha|\left(\sum_{l=0}^{j-1}2^l\right)\ <\ 2^j|\alpha| \leq 2^n|\alpha|,
  \end{align}
  as $|z|=1$.
  Plugging \eqref{eqn:nega:lambda:norm>1} into \eqref{eqn:nega:Z:norm>1} we have
  \begin{align}\label{eqn:nega:Z:estimate:norm>1}
	|Z|< \sum_{j=0}^n{n \choose j} 2^n |\alpha| = 4^n |\alpha|.
  \end{align}
  Since $\ZIsing{G}{\ybeta,z^r}\neq 0$, 
  by Lemma~\ref{lem:lowerbound:field}, 
  there exists a constant $C_{\ybeta,z^r}>1$ such that $|\ZIsing{G}{\ybeta,z^r}|>C_{\ybeta,z^r}^{-|E|}$.
  Since $|\alpha|$ is decreasing exponentially in $t$, by \eqref{eqn:nega:Z:estimate:norm>1},
  we may pick an integer $t$ that is polynomial in $n$ (and sufficiently large with respect to~$ \cst'$) such that
  \begin{align}
    |Z|< \frac{\cst'-1}{\cst'} C_{\ybeta,z^r}^{-|E|} < \frac{\cst'-1}{\cst'} |\ZIsing{G}{\ybeta,z^r}|.
	\label{eqn:nega:add:error123}
  \end{align}
  By the definition of $|Z|$ and again the triangle inequality we get
  \begin{align*}
	\frac{1}{\cst'}  =
	1 - \frac{\cst'-1}{\cst'} \leq
    \frac{|\ZIsing{G'}{\ybeta,z}|}{|w^{2t}+z|^{n(r-1)}|\ZIsing{G}{\ybeta,z^r}|} \leq 
	1 + \frac{\cst'-1}{\cst'} \leq
	 \cst'.  
   \end{align*}
   This finishes the proof.
\end{proof}

We will show how to implement a real edge interaction in the next lemma.
Unless the norm of the new interaction is $1$,
the hardness holds due to the previous two lemmas.
The failure cases are indeed polynomial-time computable.

\begin{lemma}  \label{lem:root:field:general}
  Let $\cst>1$.
  Let $\ybeta$ and $z$ be two roots of unity such that $y\not\in\{1,-1,\I,-\I\}$ and $z\not\in\{1,-1\}$.
  Then \CNonzeroNIsing{\cst}{\ybeta,z} is \numP-hard.
\end{lemma}
\begin{proof}
  Let $\ybeta=e^{\I\theta}$ and $z=e^{\I\varphi}$ and $\theta,\varphi\in[0,2\pi)$.
  Then $\theta \not\in\{0,\pi/2,\pi,3\pi/2\}$
  and $\varphi \not\in\{0,\pi\}$.

  Since $\ybeta$ is a root of unity, there exists an integer power of $\ybeta$ that equals $\ybeta^{-1}$.
  Hence we can implement $\ybeta^{-1}$ by thickenings.
  Then we implement a real interaction $w(\theta,\varphi)$ by the following gadget.
  We replace every edge by two parallel gadgets: 
  one is a $2$-stretch with interaction $\ybeta$ (on both edges) and the other is also a $2$-stretch but with $\ybeta^{-1}$.
  Then we calculate the effective edge interaction.
  When both endpoints are assigned $0$, the contribution is $(\ybeta^2+z)(1/\ybeta^2+z)=1+z^2+z(\ybeta^2+1/\ybeta^2)$.
  When both endpoints are assigned $1$, the contribution is $(\ybeta^2z+1)(z/\ybeta^2+1)=1+z^2+z(\ybeta^2+1/\ybeta^2)$ as well.
  When one endpoint is assigned $0$ and the other $1$, the contribution is $\ybeta(1+z)\cdot (1+z)/\ybeta=(1+z)^2$.
  Hence effectively on this edge the interaction is of the Ising type and its weight is $w(\theta,\varphi)=\tfrac{1+z^2+z(\ybeta^2+1/\ybeta^2)}{(1+z)^2}$.

  We claim $w(\theta,\varphi)\in\RR$.
  This is because
  \begin{align*}
	w(\theta,\varphi) & = \frac{1+z^2+z(\ybeta^2+1/\ybeta^2)}{(1+z)^2} = 1 + \frac{z(\ybeta^2+1/\ybeta^2-2)}{(1+z)^2}\\
	& = 1 + \frac{(\ybeta-1/\ybeta)^2}{z+1/z+2} = 1 + \frac{-4\sin^2\theta}{2\cos\varphi+2}\\
	& = 1 - \frac{\sin^2\theta}{\cos^2\tfrac{\varphi}{2}}.
  \end{align*}
  Notice that $\cos\tfrac{\varphi}{2}\neq 0$ as $\varphi\neq 0,\pi$.
  If $|w|<1$, then we are done by combining Lemma~\ref{lem:field:<1norm} and Corollary~\ref{cor:norm1:hard}.
  Otherwise if $|w|>1$, the lemma follows from Lemma~\ref{lem:field:>1norm} by powering $z$ to $1$, and Corollary~\ref{cor:norm1:hard}.

  The failure case is $|w(\theta,\varphi)|=1$ and hence $\sin^2\theta=2\cos^2\tfrac{\varphi}{2}$ or $\sin\theta= 0$.
  Note that $\sin\theta= 0$ implies $\ybeta=\pm1$ which contradicts our assumption.
  It is easy to implement $\ybeta^2$, which has argument $2\theta$.
  We then repeat the construction.
  If $|w(2\theta,\varphi)|\neq1$, then it is reduced to previous cases.
  Otherwise $|w(2\theta,\varphi)|=1$, implying that $\sin^22\theta=2\cos^2\tfrac{\varphi}{2}=\sin^2\theta$ or $\sin2\theta= 0$.
  The latter case is impossible as $\theta\not\in \{0,\pi/2,\pi,3\pi/2\}$.
  Hence $\sin^22\theta=\sin^2\theta$.
  It is easy to show that $\theta\in\{\pi/3,2\pi/3,4\pi/3,5\pi/3\}$ as $\theta\neq 0,\pi$.
  Therefore $2\cos^2\tfrac{\varphi}{2}=\sin^2\theta=3/4$.
  However $\cos^2\tfrac{\varphi}{2}=3/8$ has no solution $\varphi$ that is a rational fraction of $\pi$,
  which contradicts the fact that $z$ is a root of unity.
  This finishes the proof.
\end{proof}

\begin{lemma}  \label{lem:root:field:I}
  Let $\cst>1$.
  Let $\ybeta=\pm\I$ and $z$ be a root of unity that is not one of $\{1,-1,\I,-\I\}$.
  Then \CNonzeroNIsing{\cst}{\ybeta,z} is \numP-hard.
\end{lemma}
\begin{proof}
  Let $\ybeta=e^{\I\theta}$ and $z=e^{\I\varphi}$ where $\theta,\varphi\in[0,2\pi)$.
  As $\ybeta=\pm\I$, we have $\theta\in\{\pi/2,3\pi/2\}$ and $z\not\in\{1,-1,\I,-\I\}$ implies $\varphi\not\in\{0,\pi/2,\pi,3\pi/2\}$.
  We use the same $w(\theta,\varphi)\in\RR$ construction as in the proof of Lemma~\ref{lem:root:field:general}.
  If $|w(\theta,\varphi)|=0$ then $\cos^2\tfrac{\varphi}{2}=1$.
  This implies $\varphi/2\in\{0,\pi\}$ contradicting $\varphi\not\in\{0,\pi/2,\pi,3\pi/2\}$.
  If $|w(\theta,\varphi)|=1$ then $\cos^2\tfrac{\varphi}{2}=1/2$.
  This implies $\varphi/2\in\{\pi/4,3\pi/4,5\pi/4,7\pi/4\}$ also contradicting $\varphi\not\in\{0,\pi/2,\pi,3\pi/2\}$.
  Hence we can implement a real edge interaction $w(\theta,\varphi)$ such that $|w(\theta,\varphi)|\neq 0,1$.
  
  Note that $w(\theta,\varphi)=1 - \frac{\sin^2\theta}{\cos^2\tfrac{\varphi}{2}}=1 - 1/\cos^2\tfrac{\varphi}{2}<0$.
  If $w(\theta,\varphi)\in(-1,0)$, then we adopt the construction in the proof of Lemma~\ref{lem:field:<1norm} to cancel the external field of $z$.
  Hence we can reduce \CNonzeroNIsing{\cst}{w(\theta,\varphi)} to \CNonzeroNIsing{(\cst\cst')}{\ybeta,z} for any constant $\cst'>1$.
  The \numP-hardness follows from Corollary~\ref{cor:realnegy}.

  Otherwise $w(\theta,\varphi)\in(-\infty,-1)$, then we use Lemma~\ref{lem:field:>1norm} to power up the external field of $z$.
  Instead of powering $z$ to $1$, we would like to pick a positive integer $r$ such that $w(\theta,r \varphi)\in(-1,0)$, which reduces to the previous case.
  This is equivalent to $\tfrac{1}{2}<\cos^2\tfrac{r\varphi}{2}<1$, which, in turn, is equivalent to $r\varphi\in(0,\pi/2)\cup(3/2\pi,2\pi)$ modulo $2\pi$.
  Suppose $\varphi=\tfrac{2a\pi}{b}$ where $a,b$ are two co-prime positive integers and $b=3$ or $b\ge 5$ since $z\not\in\{1,-1,\I,-\I\}$.
  Assume $b\ge 5$ first. 
  As $a,b$ are co-prime, there exist two integers $l_1$ and $l_2$ such that $l_1 a+l_2 b=1$ and $l_1>0$.
  Let $r=l_1$ and we have $r\varphi/2=\tfrac{2a l_1 \pi}{b}=\tfrac{2 \pi}{b}-2 l_2 \pi$.
  This choice of $r$ meets the requirement since $\tfrac{2 \pi}{b}\in(0,\pi/2)$.

  The case left is when $b=3$, in which case $\varphi\in\{2\pi/3,4\pi/3\}$.
  We reduce \CNonzeroNIsing{\cst}{\ybeta,-z} to \CNonzeroNIsing{\cst}{\ybeta,z}.
  This suffices due to $\arg(-z)=\varphi+\pi$, which is one of the previous cases.

  Suppose $G=(V,E)$ is an input to \CNonzeroNIsing{\cst}{\ybeta,-z}.
  Introduce a new vertex $v'$ for each vertex $v\in V$.
  Since $\ybeta=\pm\I$, there exists a positive integer $t$ such that $\ybeta^t=-1$.
  Connect $v$ and $v'$ by $t$ many new edges.
  We can calculate that the effective field of $v$ in the new graph (with respect to interaction $y$ and field $z$) is $\tfrac{z-z^2}{z-1}=-z$.
  This finishes our proof.
\end{proof}

We can now prove our main theorem about this model.
 
\setcounter{counter:save}{\value{theorem}}
\setcounter{theorem}{\value{counter:field}}
\begin{theorem}
 Let $\cst>1$.
 Let $\ybeta$ and $z$ be two roots of unity.
 Then the following holds:
 \begin{enumerate}
   \item If $\ybeta=\pm\I$ and $z\in\{1,-1,\I,-\I\}$, or  $\ybeta=\pm 1$, then \ZIsing{-}{\ybeta,z} can be computed exactly in polynomial time.
   \item Otherwise \CNonzeroNIsing{\cst}{\ybeta,z} is \numP-hard.
 \end{enumerate}
\end{theorem}
\setcounter{theorem}{\value{counter:save}}
\begin{proof}
  If $y=\pm 1$, then we can replace every edge interaction by two unary constraints.
  Hence the problem is tractable for any external field.
  Consider next the case where $y=\pm\I$. If $z\in\{1,-1,\I,-\I\}$, the algorithm is from \cite{CLX14}.
  Otherwise, the hardness is from Lemma~\ref{lem:root:field:I}.  
  Finally,  for the rest of the proof, we consider the case where $y\not\in\{1,-1,\I,-\I\}$.
  For $z=1$, the hardness follows from Corollary \ref{cor:norm1:hard}.
  For $z=-1$, the hardness is obtained by combining Lemma~\ref{lem:field:-1} and Corollary \ref{cor:norm1:hard}.
     Otherwise $z\not\in\{1,-1\}$, and the hardness follows from Lemma~\ref{lem:root:field:general}.
\end{proof}

\section*{Acknowledgements}

We thank Dan Shepherd and Mark Jerrum for useful discussions.

\bibliographystyle{plain}

\bibliography{\jobname}

\end{document}